\documentclass[11pt]{article}
\usepackage{amsmath,amsfonts,amsthm,amssymb,bm}
\usepackage{graphicx,graphics,psfrag,epsfig,pstricks,geometry}
\usepackage{pstool}
\usepackage{hyperref}
\usepackage{wrapfig}
\usepackage{color}
\usepackage{soul}
\usepackage{enumitem}
\usepackage{subcaption}

\usepackage{booktabs}
\usepackage{longtable}
\usepackage{rotating}
\usepackage{array}
\usepackage{multicol}
\usepackage{multirow}
\usepackage[normalem]{ulem}

\definecolor{green}{rgb}{0.25,0.7,0.5} 
\setstcolor{red}
\newcommand{\fllist}{\begin{list}{$\bullet$}{\setlength{\parsep}{0mm}%
        \setlength{\topsep}{2mm}\setlength{\leftmargin}{6ex}%
        \setlength{\labelsep}{.5em}}\setlength{\itemsep}{1ex}}
\newcommand{\exend}{\end{list}\vspace*{1ex}}
\usepackage[numbers,sort&compress]{natbib}
\newcounter{mycount}
\newcounter{subcount}
\renewcommand{\arraystretch}{0.5}

\newcommand{\p}{\partial}
\newcommand{\sign}{{\rm sign}}

\theoremstyle{plain}  % Text in italic, and adds extra space before and after
\newtheorem{thm}{Theorem}[section]
\newtheorem{lem}[thm]{Lemma}

\theoremstyle{definition}  %  Text in roman, and adds extra space before and after
\newtheorem{definition}[thm]{Definition}

\theoremstyle{remark}  % Text in roman, but does not add extra space before or after

\newtheorem{rem}[thm]{Remark}

\title{A chemostat model with variable dilution rate due to\\ biofilm growth}
\author{Xiaochen Duan\thanks{Corresponding author, Department of Mathematics, University of Florida, Gainesville, FL 32611.  \href{mailto:duanxiaochen@ufl.edu}{duanxiaochen@ufl.edu}} ~
and ~ Sergei S. Pilyugin\thanks{Department of Mathematics, University of Florida, Gainesville, FL 32611.  \href{mailto:pilyugin@ufl.edu}{pilyugin@ufl.edu}}
}
\date{October 31, 2024}

\begin{document}
\maketitle

\begin{abstract}
In many real life applications, a continuous culture bioreactor may cease to function properly due to bioclogging which is typically caused by the microbial overgrowth. 
This is a problem that has been largely overlooked in the chemostat modeling literature, despite the fact that a number of  models explicitly accounted
for biofilm development inside the bioreactor. In a typical chemostat model, the physical volume of the biofilm is considered negligible when 
compared to the volume of the fluid. In  this paper, we investigate the theoretical consequences of removing such assumption. Specifically, we formulate a novel
mathematical model of a chemostat where the increase of the biofilm volume occurs at the expense of the fluid volume of the bioreactor, and as a result the corresponding dilution rate increases reciprocally. We show that our model is well-posed and describes the bioreactor that can operate in three distinct types of dynamic regimes: the washout equilibrium, the coexistence equilibrium, or a transient towards the clogged state which is reached in finite time. We analyze the multiplicity and the stability of the corresponding equilibria.  In particular, we delineate the  parameter combinations for which the chemostat never clogs up and those for which it clogs up in finite time. We also derive criteria for microbial persistence and
extinction. Finally, we present a numerical evidence that a multistable coexistence in the chemostat with variable dilution rate is feasible.
\end{abstract}

\section{Introduction}
In the early 20th century, the need for more efficient methods for large-scale chemical and biological reactions led to the design of the Continuous Stirred Tank Reactor (CSTR) - an apparatus where a continuous flow of reactants and products is maintained, where the contents of the reactor vessel are continuously stirred  to ensure homogeneous concentrations
of all reactants.\cite{Foutch2003} One of the main features of the classical CSTR is the tendency of the chemical reactants to achieve constant concentrations, hence the CSTR has been traditionally dubbed a chemostat. A formal description of a chemostat was published in the early 1950s by Novick and Szilard. \cite{Novick1950} Their work facilitated experiments that required controlled and sustained growth conditions, such as a constant supply of nutrients, concentration levels, temperature, pH, and oxygen levels. Over the decades, the chemostats have undergone refinement and adaptation to suit diverse biological and industrial applications ranging from macro- to microfluidic spatial scales.

One of the very early success stories in mathematical biology was the development of a mathematical theory of microbial competition for a single nutrient in the chemostat. This theory was based on the resource uptake model pioneered by J. Monod \cite{Monod1949} and others and was further developed by Hsu, Hubbell, and Waltman \cite{Hsu1977}.

In 1999, Pilyugin and Waltman proposed a simple chemostat model of the wall-attachment on micro-organisms. \cite{Pilyugin1999} In their model, $u$ represents the density of species in the channel, referred to as interior species, and $w$ represents the density of adherent species on the wall, referred to as wall-attached species. The equations for the simple chemostat model with wall growth were expressed as the following system
\begin{equation}
\label{1999.model}
     \left\{
     \begin{aligned}
        \dot{S} & = D(S^0-S) - f(S) u - g(S) w, \\
        \dot{u} & = (f(S) - D) u -\alpha u + \beta w, \\
        \dot{w} & = g(S)w -  \beta w + \alpha u,  
    \end{aligned} 
    \right.
\end{equation}
%\begin{wrapfigure}{r}{5cm}
%    \centering
%        \includegraphics[width = 0.3\textwidth, height = 0.27\textwidth]{Dilutionchange/Graph - 1.pdf}
%        \caption{The sketch of the model (\ref{1999.model})}
%\end{wrapfigure}
where $\alpha$ is the rate of adhesion, $\beta$ is the sloughing rate of wall-attached species, and $D$ is the dilution rate.
The functions $f(S)$ and $g(S)$ are the growth rates for the species $u$ and $w$ respectively. They must satisfy the
following properties:
\begin{equation*}
    \begin{aligned}
         & f'(S) > 0,~~ f \in C^1,~~ f(0) = 0 \\
         & g'(S) > 0,~~ g \in C^1,~~ g(0) = 0.
    \end{aligned}
\end{equation*}
A traditional choice for $f(S)$ and $g(S)$ is a Holling type II (also known as Monod or Michaelis-Menten) response function:
\begin{equation*}
    f(S) = \frac{m_f S}{a_f + S}, ~~ g(S) = \frac{m_g S}{a_g + S},
\end{equation*}
where, $m_f$ and $m_g$ are the maximal growth rate (units are 1/t), and $a_f$ and $a_g$ are the half-saturation constants with units of concentration \cite{MR1315301}. 

In the traditional wall-growth model, the bacteria that are attached to the wall are considered immobile, so they do not wash out from the chemostat as the contents of
the reactor are constantly diluted. An alternative way to model the differential removal of bacteria is to consider the flocculation mechanism: planktonic bacteria
may aggregate to form flocs that are less mobile due to their increased size, and therefore the flocs may exhibit a slower removal rate than the planktonic bacteria.
The flocculation has been shown to successfully explain the bacterial coexistence in the chemostat \cite{Haegeman2008,FelikhSalem2013}.

While Pilyugin and Waltman's chemostat model is useful in understanding the interaction between the wall of the chemostat and the experimental species, the theory does not show how the volume of the wall species affects the dilution rate. Actually, one effect of wall species is bioclogging. \cite{Boswell2004} Bioclogging causes of the decrease in hydraulic conductivity directly. \cite{Abdel2010} In detail, bioclogging by microbial cell bodies and their synthesized byproducts such as extracellular polymeric substance \cite{Jiang&Matsumoto1995}, which forms biofilm \cite{Taylor1990} or microcolony aggregation. \cite{Seki2001} It is shown in the first and second graph of figure \ref{modelling}.

In our model, we consider the effect of bioclogging with the variation of the volume. The dilution rate is the rate at which fresh medium is added to the reactor relative to the volume of the reactor. It is a critical parameter that controls the growth rate of microorganisms, thus it is essential to understand how wall species impact the dilution rate to maintain a controlled environment within the chemostat.

For mathematical tractability we neglect the complications associated with biofilm's geometry and simply assume that the sum of the volumes of the biofilm ($V_w$) and the fluidic chamber ($V_f$) is conserved ($V_w + V_f = V_0$), so that biofilm growth proportionally  reduces the latter volume. We introduce the limiting quantity $w_0$ of the biofilm that is sufficient to fill the entire reactor chamber. In normal (unclogged) operational regime, we have $0 \leq w <w_0$, and when $w = w_0$, we consider the chemostat completely clogged. 
 
%The situation of the interior of the reactor is shown as follows: 

\

\begin{figure}[h]
    \centering
    \includegraphics[width = 1\textwidth, height = 0.35\textwidth]{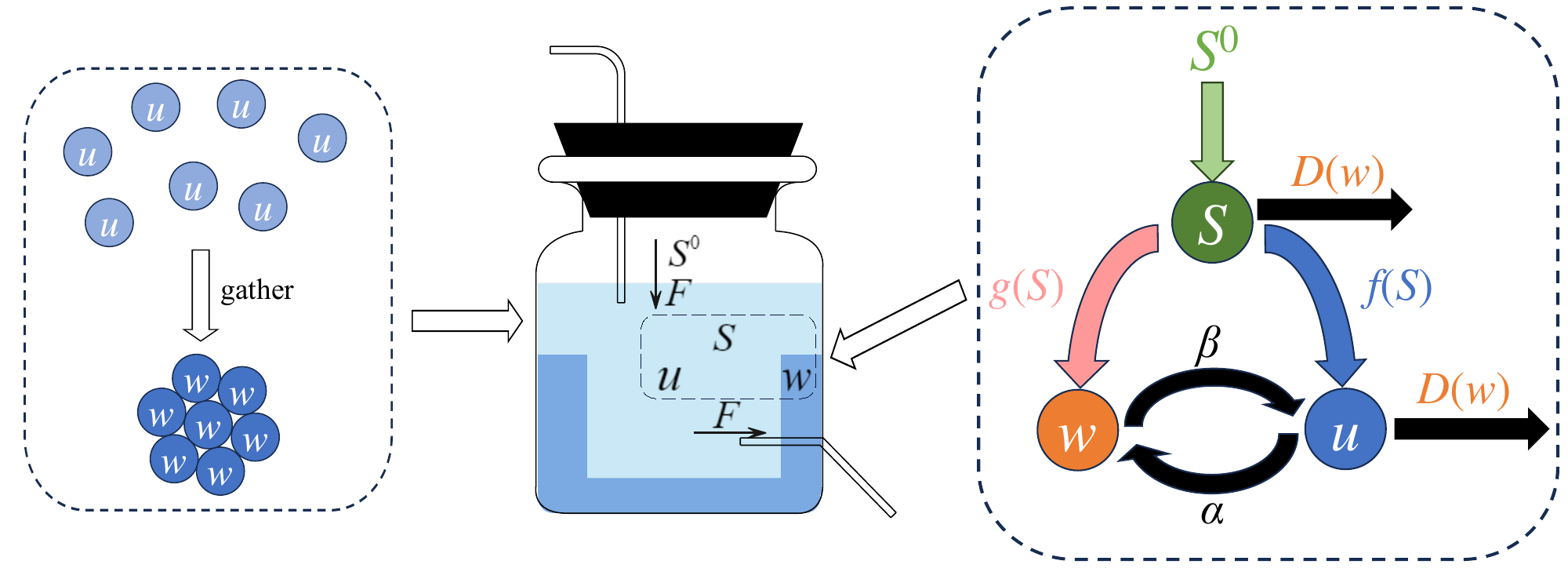}
    \caption{The aggregations of the species $u$ form the larger and heavier biopolymer which is not influenced by the dilution. Hence, these polymers will stay on the wall as the adherent, which is corresponding to the shaded part of the second figure.
    The third figure describes the relationships between all variables in the chemostat model and the dilution rate is influenced by the growth of wall-attached species $w$.}
    \label{modelling}
\end{figure}

%\begin{figure}[htbp]
%    \centering
%    \includegraphics[width = 0.33\textwidth, height = 0.46\textwidth]{Reactor.pdf}
%    \caption{The sketch of reactor}
%    \label{fig:Chemostat}
%\end{figure}

Where, the volume of the shaded part(c.f. the second figure of the figure \ref{modelling}) is $V_w$, and the light part is $V_f$. 
The assumption of proportional  volume decrease can be expressed as $\displaystyle\frac{w}{w_0} = \frac{V_w}{V}$, hence we have that
\begin{equation}
    \frac{V_f}{V_0} + \frac{w}{w_0} = 1, \quad V_f = V_0\left(1-\frac{w}{w_0}\right).
\end{equation}

In the chemostat literature, the dilution rate $D$ is defined as the ratio of $F$, the volumetric flow rate (units are length$^3/t$), to the volume of the fluidic chamber $V_f$ (i.e. $D = \dfrac{F}{V_f}$). Thus, using the relationship of volume in correspondence to wall species $w$, the dilution rate can be rewritten as
\begin{equation}\label{dilution}
    D(w) =  \frac{F}{V_f} = \frac{F}{V_0} \frac{1}{1-\displaystyle\frac{w}{w_0}} = \frac{D_0 w_0}{w_0 - w},
\end{equation}
where $D_0 = \dfrac{F}{V_0}$ is the dilution rate in the absence of the biofilm. Assuming that the flow rate $F>0$ is fixed, the dilution rate $D(w)$ is an increasing function of $w$ such that
$D(0)=D_0>0$ and $\lim_{w\uparrow w_0} D(w)=+\infty.$

Substituting (\ref{dilution}) into (\ref{1999.model}), we obtain a chemostat model with variable dilution rate due to biofilm growth:
\begin{equation}
\label{dilution.ode}
     \left\{
     \begin{aligned}
        \dot{S} & = D(w)(S^0-S) - f(S) u - g(S) w,\ S(0) \geq 0, \\
        \dot{u} & = (f(S) - D(w)) u -\alpha u + \beta w,\ u(0) \geq 0, \\
        \dot{w} & = g(S)w -  \beta w + \alpha u, \ w_0>w(0) \geq 0, 
    \end{aligned} 
    \right.
\end{equation}
where $D(w)$ is the dilution rate function defined by (\ref{dilution}) and the remaining parameters are simply inherited from the model (\ref{1999.model}).  
The following table contains the model parameters, all parameters are positive.

\begin{center}
\renewcommand{\arraystretch}{1.2}
    \begin{tabular}{cl}
        \hline
        Symbol & Description  \\
        \hline
        $t$ & Time \\
        $S$ & Concentration of nutrient in the reactor \\
        
        $u$ & Biomass concentration of species in the fluid \\        
        
        $w$ & Biomass concentration of wall-attached species\\ 
        
        $w_0$ & The maximum capacity of the wall-attached species \\
        
        $f(S)$ & Specific growth rate of species in the fluid \\
        
        $g(S)$ & Specific growth rate of wall-attached species \\

        $S^0$ & Input concentration of the nutrient  \\

        $V_0$ & Total volume of the reactor  \\

        $F$ & Volumetric flow rate of the chemostat  \\
        
        $D_0$ & Minimal dilution rate of the chemostat ($w=0$)  \\

        $D(w)$ & Variable dilution rate of the chemostat ($0 \leq w<w_0$)  \\

        $m_f$ & Maximum growth rate of species in the fluid  \\

        $m_g$ & Maximum growth rate of wall-attached species  \\

        $a_f$ & Half-saturation constant of growth rate of species in the fluid \\

        $a_g$ & Half-saturation constant of growth rate of wall-attached species \\

        $\alpha$ & Adhesion rate of species in the fluid  \\

        $\beta$ & Sloughing rate of wall-attached species  \\
        \hline
\end{tabular}
\end{center}

The rest of this paper is organized as follows. In Section \ref{sec2}, we formulate the rescaled model and prove its well-posedness. In Section \ref{sec3}, we discuss
the boundary equilibria of the rescaled system and present sufficient conditions for the species washout.  In Section \ref{sec5}, we discuss the stability properties of
the clogged state. Section \ref{sec4} is dedicated to the uniform persistence.  Section \ref{sec6} contains the discussion of the existence, stability and multiplicity of coexistence equilibria. Section \ref{sec7} is the discussion, and it concludes this paper.

%\newpage

\section{The rescaled model and its well-posedness}\label{sec2}

For mathematical convenience, we non-dimensionalize the parameters and dependent variables and also rescale the time by $1/D_0$. The system of equations (\ref{dilution.ode}) are simplified and non-dimensional quantities are indicated below with tildes. The following table describes the model parameters.

\begin{center}
\renewcommand{\arraystretch}{1.2}
    \begin{tabular}{cc}
        \hline
        Symbol & Dimensionless quantity \\
        \hline
        $\widetilde{S}$ & $S / S^0$ \\
        $\widetilde{u}$ & $u / S^0$ \\        
	 $\widetilde{w}$ & $w / S^0$ \\ 
	$\widetilde{w_0}$ & $w_0 / S^0$ \\
	$\widetilde{f}(\widetilde{S})$ & $\frac{1}{D_0}f(S^0\widetilde{S})$ \\
	 $\widetilde{g}(\widetilde{S})$ & $\frac{1}{D_0}g(S^0\widetilde{S})$ \\
	$\widetilde{D}(\widetilde{w})$ & $\frac{1}{D_0}D(S^0\widetilde{w})$ \\
	$\widetilde{\alpha}$ & $\alpha/D_0$ \\
	$\widetilde{\beta}$ & $\beta/D_0$ \\
        \hline
\end{tabular}
\end{center}
We omit the tildes and return to the original notation:
\begin{equation}
\label{dilution.model}
     \left\{
     \begin{aligned}
        \dot{S} & = D(w)(1-S) - f(S) u - g(S) w, \\
        \dot{u} & = (f(S) - D(w)) u -\alpha u + \beta w, \\
        \dot{w} & = g(S)w -  \beta w + \alpha u.  
    \end{aligned} 
    \right.
\end{equation}
Furthermore, we introduce a new time variable $\tau$ such that 
\begin{equation}
    \frac{d}{d\tau} = \left(1-\frac{w(\tau)}{w_0}\right)\frac{d}{dt}.
\end{equation} 
As a result, we obtain a new system (\ref{rescaled.model}) which we will refer to as the \textbf{rescaled system}: 
\begin{equation}
\label{rescaled.model}
     \left\{
\begin{aligned}
    \dot{S} & = 1 - S  - \left(1 - \frac{w}{w_0}\right)[u f(S) +w g(S)], \ S(0) \geq 0,\\
    \dot{u} & = -  u+ \left(1 - \frac{w}{w_{0}}\right) \left[(f(S) - \alpha )u+\beta w \right], \ u(0) \geq 0,\\
    \dot{w} & = \left(1 - \frac{w}{w_{0}}\right)[\alpha u+ (g(S) - \beta)w], \ w_0 \geq w(0) \geq 0.
\end{aligned}
    \right.
\end{equation}

The biologically relevant domain for (\ref{rescaled.model}) is  
\begin{equation*}
    \displaystyle \Omega = \left\{(S,u,w) \in \mathbb{R}^3_+|0 \leq S \leq 1, 0 \leq u, 0 \leq w \leq w_0 \right\}.
\end{equation*}

\begin{thm}
\label{well-posedness}
    The rescaled model (\ref{rescaled.model}) is well posed. 
\end{thm}
\begin{proof}
    Define $f:\mathbb{R}^3 \rightarrow \mathbb{R}^3$, which maps $(S,u,w)$ to $(f_1(S,u,w),f_2(S,u,w),f_3(S,u,w))$, where 
    \begin{equation*}
        \begin{aligned}
            f_1(S,u,w) & = 1-S - (1 - \frac{w}{w_0})(u f(S) +w g(S)), \\
            f_2(S,u,w) & = - u+ (1 - \frac{w}{w_0})[(f(S) - \alpha )u+\beta w], \\
            f_3(S,u,w) & = (1 - \frac{w}{w_0})[\alpha u+ (g(S) - \beta)w].            
        \end{aligned}
    \end{equation*}
Due to our assumptions, $ f \in C^1$, hence by the corollary 2.2.1 in \cite{ODEbook2013Hsu}, the Initial Value Problem (\ref{rescaled.model}) admits a unique solution. Meanwhile, for all $(S,u,w) \in \Omega$, $f_1(0,u,w) = 1> 0, f_2(S,0,w) = \beta(1 - \dfrac{w}{w_0})w \geq 0$, $f_3(S,u,0) = \alpha u \geq 0$, while  $f_3(S,u,w_0) = 0$. The well-posedness then
follows by the theorem A.17 in \cite{Smith2016DynamicalSA}.
\end{proof}

Theorem \ref{well-posedness} allows to infer the well-posedness of the system (\ref{dilution.model}) since these systems have the same orbits for $0 \leq w <w_0$. The main difference between the two systems is that the original system  (\ref{dilution.model})is undefined on the set $\{w=w_0\}$ which is invariant under the rescaled system (\ref{rescaled.model}).

For the original system, we find the solutions of the original system that remain in $\overline{\Omega}$ are bounded and there is a crucial lemma as follows.
\begin{lem}
    For the rescaled system (\ref{rescaled.model}), $\forall (S(0), u(0), w(0)) \in \mathring{\Omega}$, $\displaystyle\limsup_{t \rightarrow \infty}(S + u) \leq 1 + \displaystyle\frac{\beta w_0}{4}.$  
\end{lem}
\begin{proof}
    \begin{equation*}
        \begin{aligned}
            \dot{(S + u)} & =  1- S - u  - (1-\frac{w}{w_0})wg(S) - (1-\frac{w}{w_0})\alpha u + (1-\frac{w}{w_0})\beta w \\ 
            & \leq 1- S - u + (1-\frac{w}{w_0})\beta w \\
            & = 1- S - u + \beta \frac{(w_0 - w)w}{w_0} \leq 1- S - u + \beta \frac{w_0}{4}.
        \end{aligned}
    \end{equation*}
    We conclude that $\displaystyle\limsup_{t \rightarrow \infty} (S + u) \leq 1 + \displaystyle\frac{\beta w_0}{4 }$, hence the solutions of system (\ref{rescaled.model}) are bounded.
\end{proof}

\section{The boundary equilibria}\label{sec3}
The coordinates of an equilibrium of \ref{rescaled.model}) are solutions of the following system:
\begin{equation}
\label{dilution.eq}
\left\{
 \begin{aligned}
    1 - S - \left(1 - \frac{w}{w_0}\right)(u f(S) +w g(S)) = & 0 \\
    - u+ \left(1 - \frac{w}{w_{0}}\right) \left((f(S) - \alpha )u+\beta w \right) = & 0 \\
    \left(1 - \frac{w}{w_{0}}\right)[\alpha u+ (g(S) - \beta)w] = & 0  
\end{aligned}
\right.
\end{equation}
A direct inspection of system (\ref{dilution.eq}) shows that it admits exactly two boundary solutions, namely, $E_0=(1,0,0)$ and $E_w=(1,0,w_{0})$.  The equilibrium $E_0$ is called the washout or extinction equilibrium of the system (\ref{rescaled.model}). and it is also an equilibrium of the original system (\ref{dilution.model}).  In contrast, $E_w$ is a  boundary equilibrium of the rescaled system (\ref{rescaled.model}), which corresponds
 to the clogged state, and it is a point of non-uniqueness for the system (\ref{dilution.model})  since $D(w_0)=+\infty$. We defer further discussion of $E_w$ to Section \ref{sec5}.

\subsection{The washout equlibrium}
In this subsection, we focus on the washout equilibrium $E_0$ and its stability. The Jacobian matrix at $E_0$ has the form:
\begin{equation*}
    J(E_0) = \left[\begin{array}{ccc}
        -1 & -f(1) & -g(1)  \\
        0 & A & \beta \\
        0 & \alpha  & B
    \end{array}\right ],
\end{equation*}
where $A = f(1) - 1 - \alpha$, and $B = g(1) - \beta$. All eigenvalues of $J(E_0)$ are real numbers given by
\begin{equation*}
    \begin{aligned}
        \lambda_1 & = - 1, \\
        \lambda_{2,3} & = \displaystyle\frac{(A + B) \pm \sqrt{(A - B) ^2 + 4 \alpha\beta}}{2}= \displaystyle\frac{(A + B) \pm \sqrt{(A + B) ^2 + 4 (\alpha\beta-AB)}}{2}.
    \end{aligned}
\end{equation*} 
Consequently,  $J(E_0)$ is Hurwitz if and only if $A,B<0$ and $AB>\alpha\beta$. We also observe that $\lambda_3<\lambda_2$, thus regardless of its sign, $\lambda_2$ is the principal eigenvalue of the $2\times 2$ matrix $M=\left[\begin{array}{cc}
        A & \beta  \\
        \alpha & B
    \end{array}\right]$ which is positive off-diagonal, hence there exists a positive left eigenvector $(U,W)$ of $M$ associated with eigenvalue $\lambda_2$.
We summarize this paragraph in the following Theorem.

\begin{thm}
\label{thm_e0_las}
    $E_0$ is locally exponentially stable (LES) if and only if $A,B<0$ and $AB>\alpha\beta$.
\end{thm}

We also have the following global result.

\begin{thm}
\label{thm_e0_gas}
    If $E_0$ is LES, then it attracts all solutions of (\ref{rescaled.model}) in $\mathring{\Omega}$.
\end{thm}
\begin{proof} Suppose that $\lambda_2<0$ and let $(U,W)>0$ be the associated left eigenvector of $M$. Consider the function $L(u,w)=Uu+Ww \geq 0$.
Then 
$$ \dot L = U \dot u+W \dot w = (1 - w/w_0) (U,W)
        \left[\begin{array}{cc}
        f(S) - D(w) - \alpha & \beta  \\
        \alpha & g(S) - \beta
    \end{array}\right ](u,w)^T.$$
Since $f(S) \leq f(1)$ and $g(S) \leq g(1)$ and $D(w) \geq 1$ in $\mathring{\Omega}$, it follows that
$$ \dot L \leq (1 - w/w_0) (U,W) M (u,w)^T = \lambda_2 (1-w/w_0) L.$$
Since the solutions of (\ref{rescaled.model}) are bounded, the $\omega$-limit set of our solution is a part of
the largest invariant set in $\{\dot L=0\}$ due to LaSalle's invariance principle. We have that $\dot L=0$ if $L=0$ or if $w=w_0$.
The plane $\{w=w_0\}$ is invariant under (\ref{rescaled.model}) and all solutions in that plane converge to $E_w$. The Jacobian matrix
of (\ref{rescaled.model}) at $E_w$ has the form
\begin{equation*}
    J(E_w) = \left[\begin{array}{ccc}
        -1 & 0 & g(1)  \\
        0 & -1 & -\beta \\
        0 &0  & -B
    \end{array}\right ],
\end{equation*}
where $B<0$ due to our assumptions. Therefore, $E_w$ is a saddle, and its 2-dimensional stable manifold coincides with the plane $\{w=w_0\}$.
Consequently, no solution starting in $\mathring{\Omega}$ can converge to $E_w$, hence the $\omega$-limit set of our solution
must contain points inside the set $\{\dot L=0\}$ but outside the plane $\{w=w_0\}$. Therefore, $\omega$-limit set of our solution must 
contain $E_0$, but since $E_0$ is LES, the $\omega$-limit set of our solution must coincide with $E_0$.
 
We conclude that $E_0$ attracts all solutions of (\ref{rescaled.model}) in $\mathring{\Omega}$.
\end{proof}

\begin{rem}\label{remark}
 If we denote the vector field of (\ref{dilution.model}) by $F$, then the vector field of (\ref{rescaled.model}) is given by $\tilde{F}=(1-w/w_0)F$.
For any point $E\in\Omega \cap \{ w<w_0\}$, we have that $F(E)=0 \Leftrightarrow  \tilde{F}(E)=0,$  thus systems (\ref{dilution.model}) and (\ref{rescaled.model})
have the same set of equilibria in the set $\Omega \cap \{ w<w_0\}.$
Furthermore, if $F(E)=\tilde F(E)=0$, then
$$ \frac{\p \tilde F}{\p x}(E)=(1-w/w_0) \frac{\p F}{\p x}(E),$$
hence the eigenvalues of the Jacobian matrices of both systems have the same sign structure and these matrices possess identical eigenspaces. 

The main difference between the two systems is that the rescaled system (\ref{rescaled.model}) defines a complete flow on the set $\Omega$,
while some solutions of the original system (\ref{dilution.model}) may be defined on bounded time intervals only as we will show in Section \ref{sec5}.
This is the primary reason why we  study the extinction, the persistence and other global properties in terms of the flow generated by the solutions of the rescaled system (\ref{rescaled.model}).
\end{rem}

\subsection{Numerical results for the washout equilibrium}
In this subsection, we present some numerical simulations to illustrate  the theorems \ref{thm_e0_las} and \ref{thm_e0_gas}.  To do so, we choose a set of parameters which satisfy the condition of Theorem \ref{thm_e0_las} and simulate a solution with a randomly chosen initial condition in $\Omega.$

%In the plot below, the three graphs represent the sample path image for the original system (\ref{dilution.model}). The green curve, red curve and blue curve are corresponding to $S$, $u$ and $w$, respectively.

\begin{figure}[htb]
    \centering
    \subfloat{\includegraphics[width=0.48\linewidth]{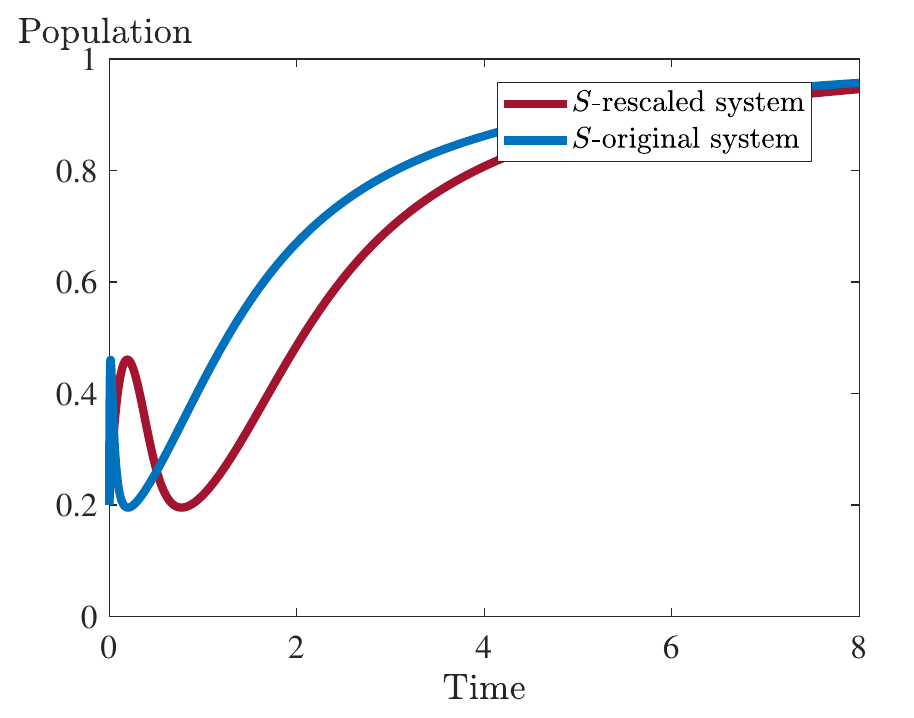}}
    \hfill
    \subfloat{\includegraphics[width=0.48\linewidth]{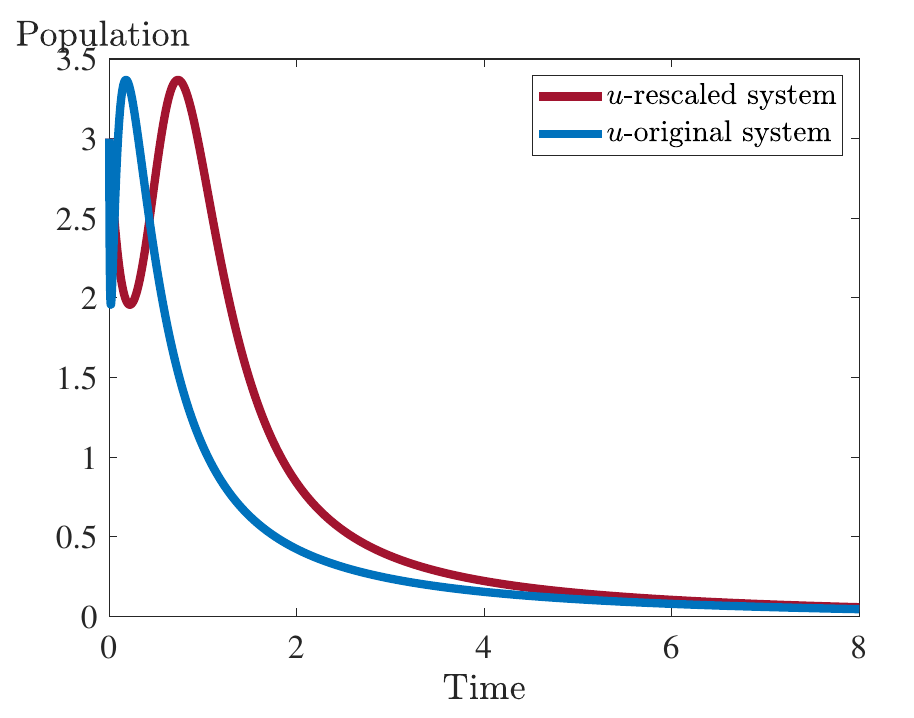}} \\
    
    \subfloat{\includegraphics[width=0.48\linewidth]{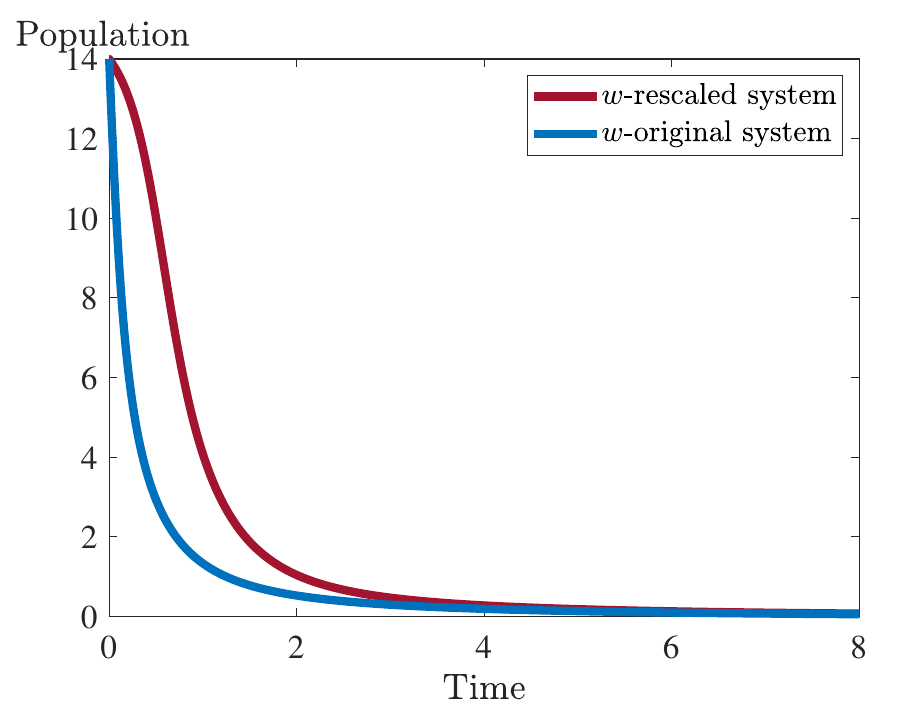}}    
    \caption{This Figure shows the simulated trajectories of the original system (\ref{dilution.model}) (shown in blue) and the rescaled system (\ref{rescaled.model}) (shown in red) corresponding to the same initial condition $(S(0), u(0), w(0)) = (0.2, 3, 14) \in \Omega$. The parameters and the growth functions are given by $\alpha = 1, \beta = 1.4, w_{0} = 14.7$, $f(S) = \dfrac{1.03S}{3.7+S}$, and $g(S) = \dfrac{0.6S}{0.8+S}$, respectively.}
    \label{washout}
\end{figure}

%\begin{figure}[htb]
%    \centering
%    \subfloat{\includegraphics[width=0.45\linewidth]{Dilutionchange/original/washout/S_t.pdf}}
%    \hfill
%    \subfloat{\includegraphics[width=0.45\linewidth]{Dilutionchange/original/washout/u_t.pdf}} \\
%    
%    \subfloat{\includegraphics[width=0.45\linewidth]{Dilutionchange/original/washout/w_t.pdf}}    
%    \caption{The simulation of the trajectory of the original system (\ref{dilution.model}). We assume that the initial value of the system (\ref{dilution.model}) is $(S(0), u(0), w(0)) = (0.2, 3, 10) \in \Omega$. The parameter setting $\alpha = 5, \beta = 4, D_{0} = 5, $ and two growth rate functions are given by $f(S) = \frac{1.03S}{3.7+S}$ and $g(S) = \frac{0.6S}{0.8+S}$.}
%    \label{washout}
%\end{figure}

\section{The clogged state} \label{sec5}
The clogged state $E_w=(1,0,w_0)$ is an equilibrium of the rescaled system  (\ref{rescaled.model}).
The corresponding Jacobian matrix has the form
\begin{equation*}
J(E_w) = \left[\begin{array}{ccc}
        -1 & 0 & g(1)  \\
        0 &  -1 &  - \beta  \\ 
        0 & 0 & - (g(1) - \beta)
    \end{array}\right ].
\end{equation*}
Hence, $E_w$ is LES if $\beta < g(1)$, and $E_w$ is a saddle point with 2-dimensonal stable manifold $\{w=w_0\}$ if $\beta>g(1)$. 

\begin{lem} \label{Ew1}
    If $E_w$ is LES, then $E_0$ is unstable.
\end{lem}
\begin{proof} If $E_w$ is LES, then $B=g(1)-\beta >0$. Consequently, regardless of the sign of $A$, the principal eigenvalue $\lambda_2$ of 
of the $2\times 2$ matrix $M=\left[\begin{array}{cc} A & \beta  \\ \alpha & B \end{array}\right]$ is positive, which implies that $E_0$ is unstable
\end{proof}

Lemma \ref{Ew1} implies that  $g(1) > \beta $, then $E_w$ has a nontrivial basin ${\mathcal B} \subset \Omega$ of attraction. In this case,
there exists on open set of initial conditions which result in chemostat clogging. The next result shows that when the clogged state is reached,
it is always reached in finite time.

\begin{lem}
\label{cloggingtime}
    If the clogged state is reached, then it is reached in finite time.
\end{lem}
\begin{proof}
Consider $\nu = g(1) - \beta > 0$, then for solutions in ${\mathcal B}$
converging to $E_w$ in the rescaled system (\ref{rescaled.model}), we will have
$$ \frac{{w_0}-w(\tau)}{w_0} \sim \exp(-\nu \tau), \quad \tau \to \infty,$$
hence in the original system 
$$ \frac{dt}{d \tau}=\frac{{w_0}-w(\tau)}{w_0}=O( \exp(-\nu \tau)), \quad \lim_{\tau \to \infty} t(\tau) \leq K \int_0^\infty e^{-\nu \tau}\, d\tau=\frac{K}{\nu}<\infty, $$
which implies that the corresponding solution of the original system (\ref{dilution.model}) will reach $E_w$ in finite time.
\end{proof}

In the remainder of this section, we present sufficient conditions for global stability of $E_w$ under the  rescaled system (\ref{rescaled.model}).

\begin{thm}
\label{gas_of_ew}
        If $f(1) -\alpha<1$ and $\beta>0$ is sufficiently small, then $E_w$ attracts all solutions of (\ref{rescaled.model}) in $\mathring \Omega$.
\end{thm}
\begin{proof} We define an auxiliary function
$$ G(b):= g^{-1}(b)+ \frac{b w_0 f(g^{-1}(b))}{4(1+\alpha -f(1))} + \frac{b w_0}{4}, \quad b\in [0,g(1)).$$
Since $f(0)=g(0)=0$ and both $f$ and $g$ are increasing functions, it follows that $G(0)=0$, $G(g(1))>1$, and $G$ is  also an increasing function.
Thus, there exists a unique $b^* \in (0,g(1))$ such that $G(b^*)=1$. It is easy to see that if $\beta<b^*$, then $\beta<g(1)$, and $E_w$ is LES. 
In what follows, we will show that if $\beta<b^*$, then $E_w$ attracts all solutions of (\ref{rescaled.model}) in $\mathring \Omega$. 

The $u$-equation in (\ref{rescaled.model}) implies that 
\begin{align*}
        \dot u & = -u +(1-w/w_0)[(f(S)-\alpha) u + \beta w]  \\  
         & \leq -(1+\alpha -f(S)) u +  (1-w/w_0)\beta w \\
         & = -(1+\alpha -f(1)) u +  \frac{\beta w_0}{4}.
\end{align*}
Thus, we have that 
$$\limsup_{t\to\infty} u(t) \leq  u^*=\frac{\beta w_0}{4 (1+\alpha -f(1))}.$$

The $S$-equation in (\ref{rescaled.model}) implies that 
\begin{align*}
        \dot S & =1-S -(1-w/w_0)[f(S)u + g(S) w]  \\  
         & \geq 1-S -f(S) u - g(S)  \frac{\beta w_0}{4} \\
         & = 1-S -f(S) u^* - g(S)  \frac{\beta w_0}{4} + f(S)(u^* -u) \\
         & = 1-\left(S+\frac{\beta w_0 f(S)}{4 (1+\alpha -f(1))} + g(S) \frac{\beta w_0}{4}\right) + f(S)(u^* -u)\\
         & > 1-\left(S+\frac{b^* w_0 f(S)}{4 (1+\alpha -f(1))} + g(S) \frac{b^* w_0}{4}\right) + f(S)(u^* -u)+g(S) \frac{(b^*-\beta) w_0}{4}.
\end{align*}
In particular, if $S<g^{-1}(b^*)$, then we have that 
$$ \dot S > 1-G(b^*)+ f(S)(u^* -u) +g(S) \frac{(b^*-\beta) w_0}{4}=f(S)(u^* -u) +g(S) \frac{(b^*-\beta) w_0}{4}.$$
Therefore, for all sufficiently large times and for any $S<g^{-1}(b^*)$, we have that
$$  \dot S \geq g(S) \frac{(b^*-\beta) w_0}{8} >0.$$
This implies that $$\liminf_{t\to\infty} S(t) \geq  g^{-1}(b^*),$$
or equivalently, that $\liminf_{t\to\infty} g(S(t)) \geq b^* >\beta.$ In other words, $g(S(t)) >\beta$ for all sufficiently large times.

Finally, the $w$-equation in (\ref{rescaled.model}) implies that 
$$ \dot w = (1-w/w_0)[\alpha u + (g(S)- \beta) w] \geq (b^*- \beta) w(1-w/w_0),$$
for all sufficiently large times. Therefore, $ \lim_{t\to\infty} w(t)=w_0$. This implies that the corresponding solution converges to $E_w$. \end{proof}

\subsection{Numerical results for the clogged state}
In this part, we will show some numerical simulations and the property we showed is supported. Meanwhile, we will give some reasonable parameters to satisfy the condition of the clogged state. 

In Figure~\ref{clogged state}, the red curve represents the sample path image of the rescaled system, the blue curve represents the sample path image of the original system. The three subgraphs represent the nutrient-time, normal speices-time, and adherent species-time relationship in turn. The lower right graph gives us the $t - \tau$ relationship, which supports the theorem \ref{cloggingtime}.

\begin{figure}[h]
    \centering
    \subfloat{\includegraphics[width=0.45\linewidth]{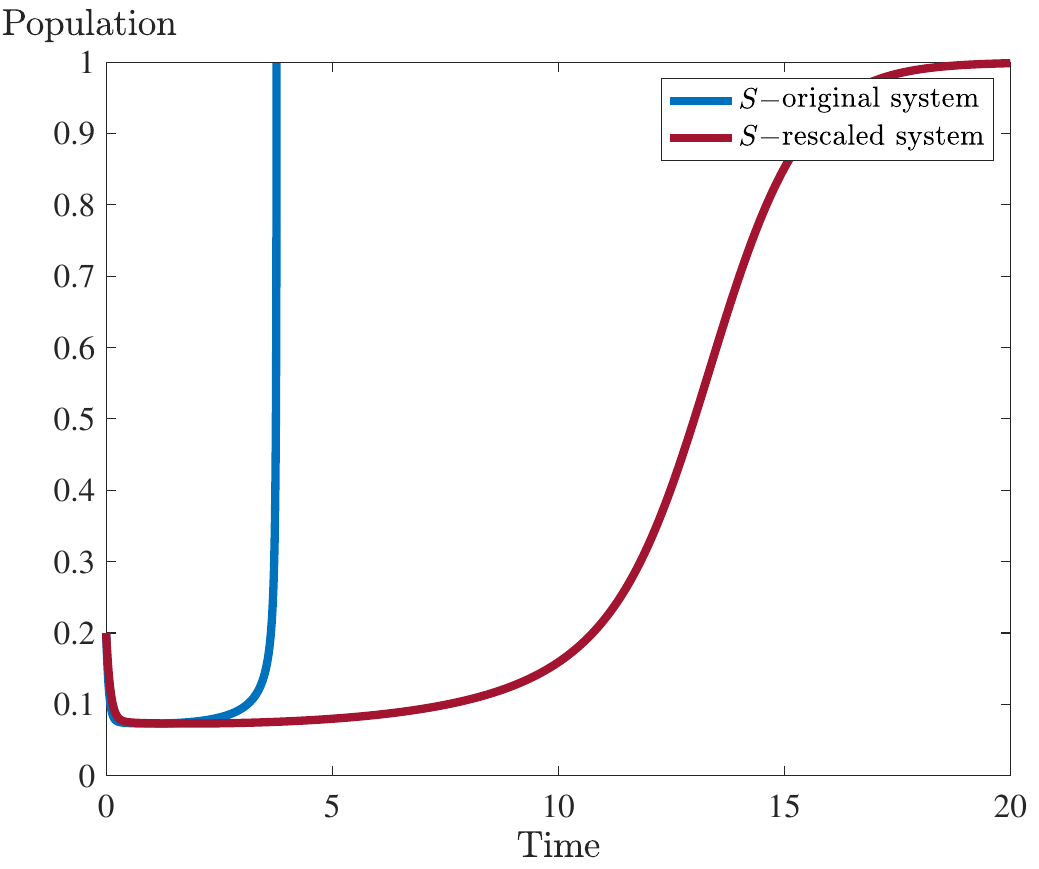}}
    \hfill
    \subfloat{\includegraphics[width=0.45\linewidth]{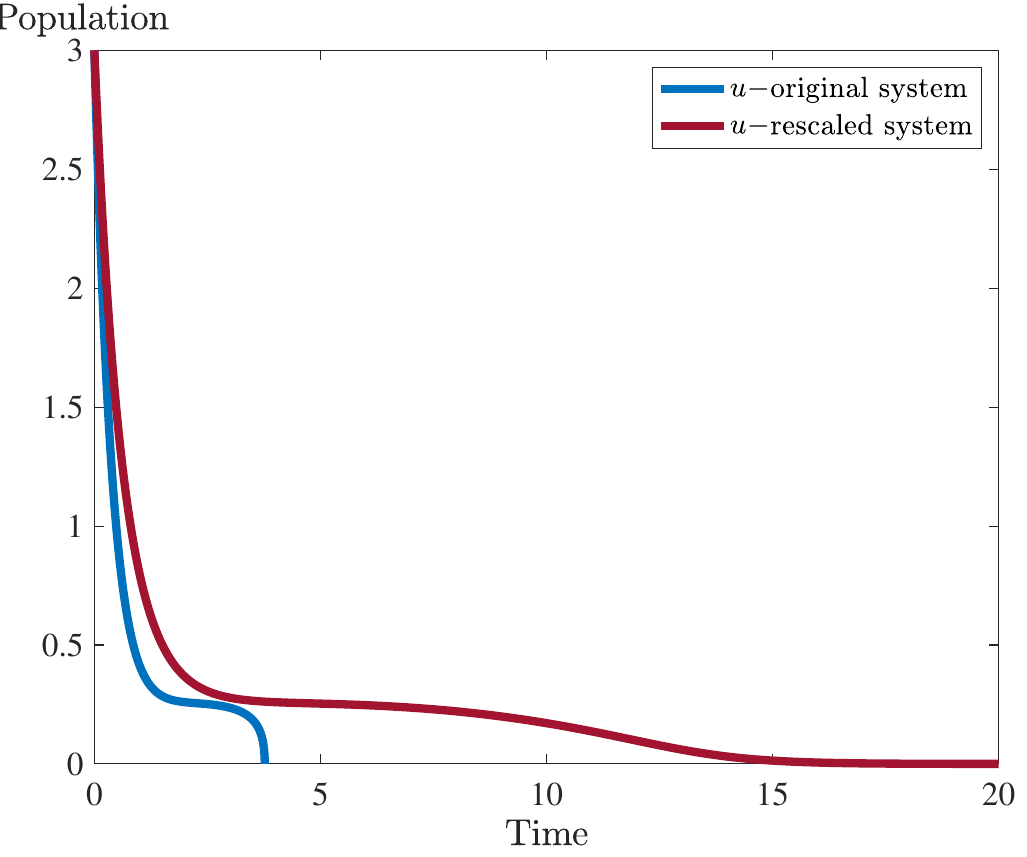}} \\
    
    \subfloat{\includegraphics[width=0.45\linewidth]{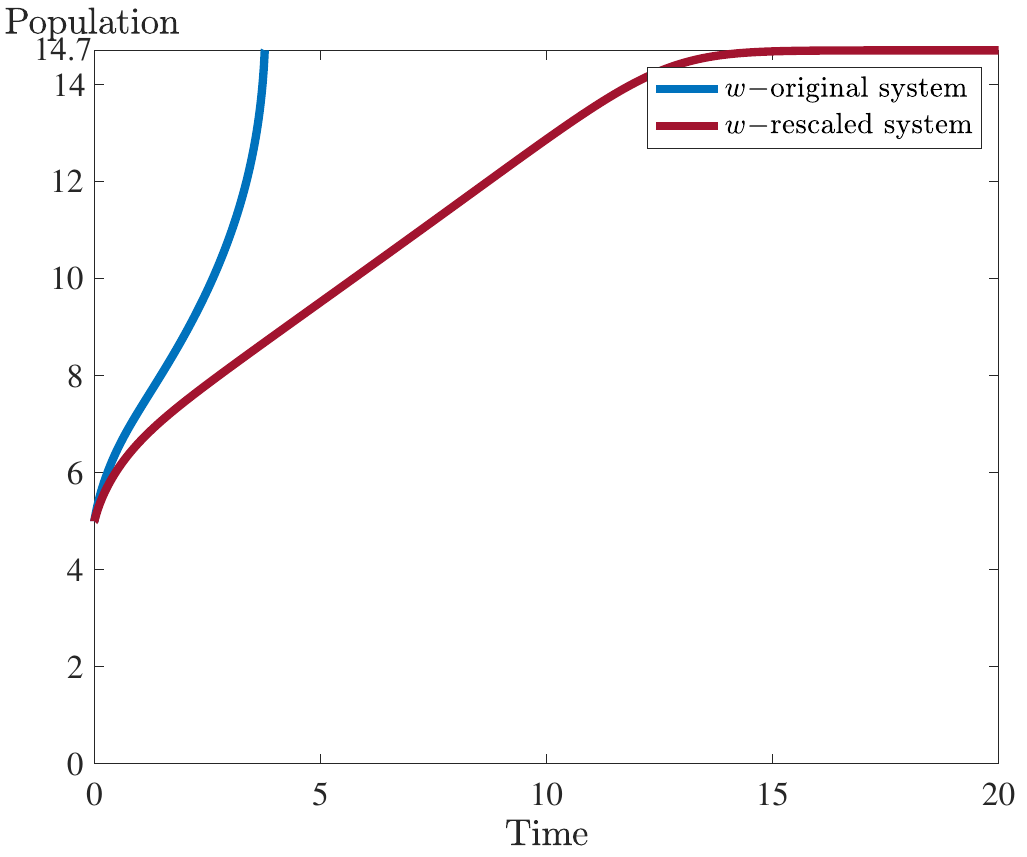}}    
    \hfill
    \subfloat{\includegraphics[width=0.45\linewidth]{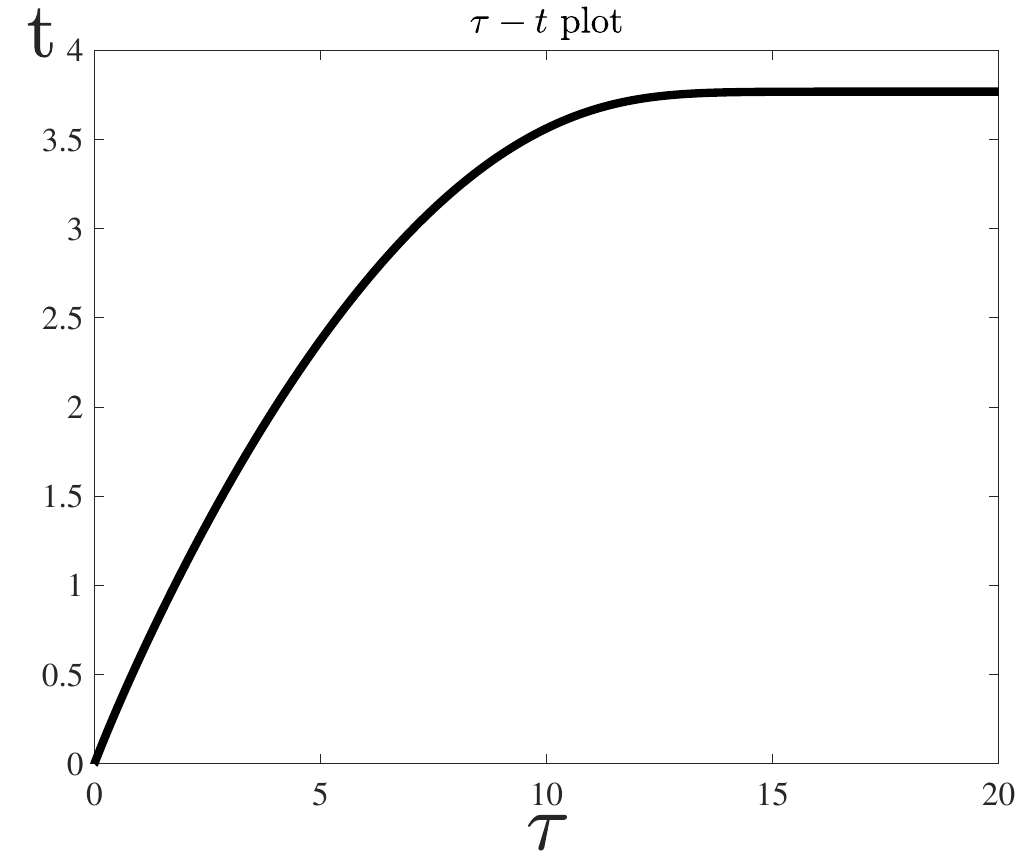}}
    \caption{This Figure shows a numerically simulated trajectory of the rescaled system (\ref{rescaled.model})corresponding to the initial condition $(S(0), u(0), w(0)) = (0.2, 3, 5) \in \Omega$. The parameters and the growth functions are given by $\alpha = 1, \beta = 0.1, w_{0} = 14.7$, $f(S) = \dfrac{1.03S}{3.7+S}$, and $g(S) = \dfrac{0.6S}{0.8+S}$, respectively. The lower right graph describes the relationship between the physical time $t$ and the rescaled time $\tau$. Numerically, the clogged time is $t_{\max} \sim 3.72$.}
    \label{clogged state}
\end{figure}

\newpage

In the theorem \ref{gas_of_ew}, we showed that $\beta > 0$ is sufficiently small implies that the global stability of $E_w$. However, we cannot show that $E_w$ is GAS if $\beta$ is not sufficiently small since there are some algebraic difficulties. We still believe $E_w$ is GAS under the hypothesis of theorem \ref{gas_of_ew}. In the plot below, this is the orbit of some different initial values.

\begin{figure}[htb]
    \centering
    \includegraphics[width=0.6\textwidth]{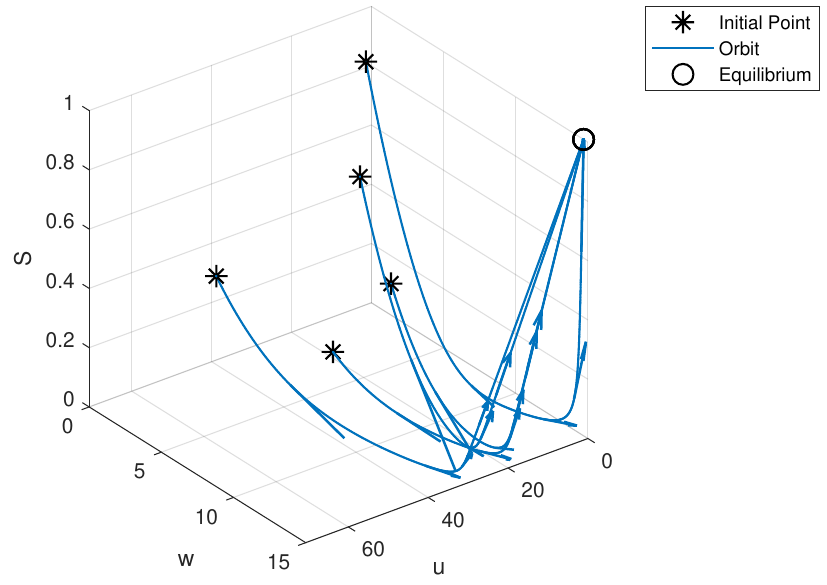}
    \caption{In this Figure, we illustrate the scenario in which the clogged state $E_w$ is globally attractive. The parameter values and the growth functions are given by  $\alpha = 1, \beta = 0.1, w_{0} = 14.7$, $f(S) = \dfrac{1.03S}{3.7+S}$, and $g(S) = \dfrac{0.6S}{0.8+S}$, respectively.  With these parameter values, the system admits no positive equilibria. Since
    $g(1) =\frac{1}{3} >\beta$, $E_w$ is stable. This Figure shows 5 orbits corresponding to the initial conditions chosen randomly from $\Omega$. All of  these orbits are attracted by $E_w$.}
    \label{fig: example}
\end{figure}

\section{Uniform persistence}\label{sec4}
Introduced by Freedman, Waltman \cite{Freedman1977257}, Gard \cite{Gard1980165,Gard198061}, Gard and Hallam \cite{Gard198061}, Hallam \cite{Hallam1978}, and Schuster, Sigmund and Wolff \cite{Schuster1978743}, the population-theoretic concept of persistence describes the long term survival of the species. Within the framework of some population model governing the dynamics of species $x$ we say that $x$ goes extinct if $\lim_{t \to\infty} x(t)=0$, and $x$ persists if  $\lim_{t \to\infty} x(t)\not=0$ provided that $x(0)>0$. Two types of persistence are distinguished: the case $\limsup _{t \to\infty} x(t)>0$ corresponds to {\bf weak persistence}, and the case $\liminf _{t \to\infty} x(t)>0$ corresponds to {\bf strong persistence}.

In this section, we aim to obtain sufficient conditions for the uniform (strong) persistence in the rescaled system (\ref{rescaled.model}), which is stated as follows:
$$\exists \epsilon > 0, \forall (S(0),u(0),w(0)) \in \mathring{\Omega}: \quad
\liminf_{t\rightarrow\infty}{S(t)} \geq \epsilon,\ \liminf_{t\rightarrow\infty}{u(t)} \geq \epsilon, \ \liminf_{t\rightarrow\infty}{w(t)} \geq \epsilon.$$
We introduce the concept of $\rho$-persistence. Let $X$ be an arbitrary nonempty set and $\rho: X \rightarrow \mathbb{R}^+$.
\begin{definition}\cite{Smith2016DynamicalSA}
     A semiflow $\Phi : \mathbb{R}^+ \times X \rightarrow X$ is called weakly $\rho$-persistent, if
\begin{equation*}
    \limsup_{t\rightarrow \infty} \rho(\Phi(t, x)) > 0 \quad \forall x \in X, \rho(x) > 0.
\end{equation*}
     $\Phi$ is called strongly $\rho$-persistent, if
\begin{equation*}
    \liminf_{t\rightarrow \infty} \rho(\Phi(t, x)) > 0 \quad \forall x \in X, \rho(x) > 0.
\end{equation*}
A semiflow $\Phi : J \times X \rightarrow X$ is called uniformly weakly $\rho$-persistent, if
there exists some $\epsilon > 0$ such that
\begin{equation*}
    \limsup_{t\rightarrow \infty} \rho(\Phi(t, x)) > \epsilon \quad \forall x \in X, \rho(x) > 0.
\end{equation*}
 $\Phi$ is called uniformly (strongly) $\rho$-persistent, if
\begin{equation*}
    \liminf_{t\rightarrow \infty} \rho(\Phi(t, x)) > \epsilon \quad \forall x \in X, \rho(x) > 0.
\end{equation*}
\end{definition}

\begin{lem}{\textbf{(Butler-McGehee lemma)} \cite{MR0848269,Wolkowicz1984}}
\label{Butler-Mcghee} Let $\Phi : \mathbb{R}^+ \times X \rightarrow X$ be a  smooth semiflow on $X \subseteq \mathbb{R}^n.$ For $x_0 \in X$, let $\omega(x_0)$
denote the omega-limit set of $x_0$.  Suppose that $p \in X$ is an isolated hyperbolic critical point of $\Phi$ such that $p \in \omega(x_0)$. Then the following dichotomy holds:
\begin{enumerate}
    \item Either $\omega(x_0) = \left\{p\right\}$;
    \item Or there exist points $p^s$ and $p^u$ in $\omega(x_0)$, with $p^s \in W^s(p) \backslash \{p\}$ and $p^u \in W^u(p) \backslash \{p\}$, where $W^s (p)$ and $W^u (p)$ denote the stable and unstable manifolds of $p$ respectively.
\end{enumerate}

\end{lem}

We recall that the eigenvalues of $J(E_0)$ are
$$ \lambda_1 = - 1, \quad  \lambda_{2,3} = \displaystyle\frac{(A + B) \pm \sqrt{(A - B) ^2 + 4 \alpha\beta}}{2}. $$ 
In Section \ref{sec3}, we proved that case $\lambda_2<0$ corresponds to extinction. Now we will prove that the case $\lambda_2>0$ corresponds to persistence. 
Throughout this section we assume that all equilibria of (\ref{dilution.model}) and (\ref{rescaled.model}) are hyperbolic.
 
Due to Remark \ref{remark}, the following two results are valid for both systems (\ref{dilution.model}) and (\ref{rescaled.model}), so it is sufficient
to prove them for system (\ref{dilution.model}).
 
\begin{lem}
\label{lem1 E_0}
    If $E_0$ is unstable, then $\mathring{\Omega} \bigcap W^s(E_0)=\emptyset.$
\end{lem}

\begin{proof} Suppose that $x_0 \in \mathring{\Omega} \bigcap W^s(E_0)$, then 
the corresponding solution $(S(t),u(t),w(t)) \in \mathring{\Omega}$ through $x_0$ converges to $E_0$ as $t \to +\infty$. 
Since $E_0$ is unstable and hyperbolic, we have that $\lambda_2 >0$ and the matrix $M$ admits a positive left eigenvector $(U,W)$.
We introduce an auxiliary function $L(u(t),w(t)) = Uu(t)+Ww(t)>0$, such that 
$$
      \dot{L} = (U,W) J(S,w) (u,w)^T = (U,W) M (u,w)^T + (U,W) (J(S,w)-M) (u,w)^T,$$ where $$ J(S,w)= \left[\begin{array}{cc}
        f(S) - D(w) - \alpha & \beta  \\
        \alpha & g(S) - \beta
    \end{array}\right ].
$$
Since $M=J(1,0)$ and the solution $(S(t),u(t),w(t))$ converges to $E_0=(1,0,0)$ as $t \to +\infty$, there exists $T \geq 0$ such that for all $t>T$,
$\| J(S(t),w(t))-M\| \leq \frac{\lambda_2}{2},$ hence
$$ \dot{L} \geq (U,W) (M -\frac{\lambda_2}{2}) (u,w)^T=\frac{\lambda_2}{2} L >0, \quad t>T,$$
which implies that $L(u(t),w(t)) \to +\infty$ as $t\to \infty$, which contradicts our assumption that $u(t),w(t) \to 0$. Therefore, $\mathring{\Omega} \bigcap W^s(E_0)=\emptyset.$ \end{proof}

\begin{lem}
\label{lem2 E_0}
    If $E_0$ is unstable, then for any $x_0 \in \mathring{\Omega},\ E_0 \notin \omega(x_0).$ 
\end{lem}

\begin{proof} If $x_0 \in \mathring{\Omega}$, then $x_0 \not\in W^s(E_0)$ by Lemma \ref{lem1 E_0}. Hence, if $E_0 \in  \omega(x_0),$ then due to Lemma
\ref{Butler-Mcghee}, there is a point $q \in W^s (E_0) \bigcap \omega(x_0)$ such that $q\not=E_0$. Since $\omega(x_0) \subset \bar\Omega$, it follows that
$q \in W^s (E_0) \bigcap \bar \Omega$.  Hence, $q=(s^*,0,0)$ for some $s^*\not=1$. This implies that the orbit $O(q)$ is unbounded, and since $\omega(x_0)$,
it implies that $O(q) \subset \omega(x_0)$, so that $\omega(x_0)$ must be also unbounded. This contradict the boundedness of solutions. Therefore, $E_0 \notin \omega(x_0).$
\end{proof}

In the remainder of this section, we use $\rho(S,u,w)=u+w$ and study the global properties of the flow of the rescaled system (\ref{rescaled.model}).

\begin{lem}
\label{persistence_of_(u+w)}
If $E_0$ is unstable, then the system (\ref{rescaled.model}) is uniformly weakly $\rho-$persistent.
\end{lem}
\begin{proof} Since $E_0$ is unstable, we know that $\lambda_2>0$. Recall that $\lambda_2$ is the principal eigenvalue of the matrix
$$ M=J(1,0)= \left[\begin{array}{cc} f(1) - D(0) - \alpha & \beta  \\  \alpha & g(1) - \beta \end{array}\right ]. $$
Consequently, the principal eigenvalue of the matrix  $\hat M:=M-\frac{\lambda_2}{2} I_{2\times 2}$ is equal to $\frac{\lambda_2}{2}>0$, moreover, the matrix $\hat M$
has the same positive principal left eigenvector $(U,W)>0$ as $M$ does, that is, $(U,W)\hat M = \frac{\lambda_2}{2}(U,W)$. By continuity, there exists $\epsilon>0$ such that for all $S>1-\epsilon$ and for all $w<\epsilon$,
$$ J(S,w)= \left[\begin{array}{cc} f(S) - D(w) - \alpha & \beta  \\  \alpha & g(S) - \beta \end{array}\right ] \geq \hat M. $$
Let $$\delta=\min\left(\epsilon, \frac{ \epsilon}{\max(f(1),g(1))},\frac{w_0}{2} \right)>0.$$
Now suppose that there exists a solution $(S,u,w) \in \Omega$ of (\ref{rescaled.model}) such that $u(0)+w(0)>0$ and $\limsup_{t \to \infty} (u(t)+w(t)) < \delta.$
Then for all sufficiently large times, we have that $u(t)+w(t) <\delta \leq \epsilon$, hence $1/2<1-w(t)/w_0 < 1$, and
$$ \dot S =1-S - (1-w/w_0)[u f(S) + w g(S)] \geq 1-S - \delta \max(f(1),g(1)) > 1-\epsilon -S,$$
which implies that there exists some $T \gg 1$ such that $S(t) >1 -\epsilon$ and $w(t)<\epsilon$ for all $t>T$.

Using the auxiliary function $L=Uu(t)+Ww(t)$, we find that
$$ \dot L=(1-w/w_0)(U,W)J(S,w)(u,w)^T \geq \frac{1}{2}(U,W) \hat M (u,w)^T = \frac{\lambda_2}{4}L>0,$$
for all $t>T$. This implies that $L \to \infty$ as $t\to \infty$, which contradicts the inequality $\limsup_{t \to \infty} (u(t)+w(t)) < \delta.$

Therefore, for any solution $(S,u,w) \in \Omega$ with $u(0)+w(0)>0$, we must have that $\limsup_{t \to \infty} (u(t)+w(t)) \geq \delta.$ This concludes the
proof of uniform weak $\rho-$persistence of (\ref{rescaled.model}). \end{proof}

\begin{thm}
\label{rho-strong-persistent}
If $E_0$ is unstable, then the system (\ref{rescaled.model}) is uniformly strongly $\rho-$persistent.
\end{thm}
\begin{proof}
We apply the Theorem 4.13 in \cite{Smith2016DynamicalSA} to prove this theorem. There are four conditions to verify: 

\begin{enumerate}
\item[$\hat{\heartsuit}_0$] $X$ is a metric space, $\rho$ is uniformly continuous, $\sigma = \rho \circ \phi: \mathbb{R}^+ \times X \to \mathbb{R}^+$ is continuous. 
    \item[$\hat{\heartsuit}_1$] There exists a nonempty $B \subseteq X$, such that $\forall x \in X$, $\rho(x) > 0 \implies d(\phi_t(x), B) \rightarrow 0.$
    \item[$\hat{\heartsuit}_2$] If $0 < \epsilon_1 < \epsilon_2 < \infty$, then $B \bigcap \left\{\epsilon_1 \leq  \rho(x) \leq \epsilon_2\right\}$ is compact.
    \item[$\hat{\heartsuit}_3$] There are no $y \in B \subset X$ ($B$ is non-empty) and  $s,t \in \mathbb{R}^+$ such that $\rho(y) > 0, \sigma(t,y) = 0$ and $\sigma(t+s,y) > 0.$
\end{enumerate}

Recall $\displaystyle \Omega = \left\{(S,u,w)|\text{ }0 \leq S \leq 1, 0 \leq u \leq 1 +  \frac{\beta w_0}{4}, 0 \leq w \leq w_0 \right\}$, and define $X$ and $B$ so that $X =B= \Omega$. By Lemma \ref{persistence_of_(u+w)}, there exists $\eta > 0$ such that $\limsup_{t \rightarrow \infty}\rho(\phi_t(x_0)) > \eta$, that is, $\phi$ is uniformly weakly $\rho-$persistent. 

To complete the proof, it suffices to verify the conditions ($\hat{\heartsuit}$). First of all, $\hat{\heartsuit}_0$ is satisfied by the definition of $X, \rho$ and $\sigma$. By the well-posedness and boundedness of the rescaled system (\ref{rescaled.model}), $X=B$ is forward invariant, hence $\hat{\heartsuit}_1$ is trivially satisfied.

For $\hat{\heartsuit}_2$, we note that the set $B \bigcap\left\{\epsilon_1 \leq  \rho(x) \leq \epsilon_2\right\}$ is a closed subset of a bounded set $B \subset \mathbb{R}^3$,
thus  $B \bigcap \left\{\epsilon_1 \leq  \rho(x) \leq \epsilon_2\right\}$  is compact.

The condition $\hat{\heartsuit}_3$  is satisfied since $u+w=0$ implies that $u=w=0$, which implies that $u(t)=w(t)=0$ for all $t$. 

Therefore, the system (\ref{rescaled.model}) is uniformly strongly $\rho-$persistent by Theorem 4.13 in \cite{Smith2016DynamicalSA}.
\end{proof}

Theorem \ref{rho-strong-persistent} implies that if  $E_0$ is unstable, then there exists $\gamma > 0$ such that for any initial value 
$x_0 \in \mathring{\Omega}$, $\liminf_{t \rightarrow \infty}(u + w) \geq \gamma$.

\subsection{Persistence against clogging}

In this subsection, we will study the $r$-persistence of the rescaled system (\ref{rescaled.model}) for $r = w_0 - w$. This type of persistence indicates that
the solutions of (\ref{rescaled.model}) stay away from the boundary set $\{w=w_0\}$ which corresponds to a clogged state of the chemostat. 
In other words, when (\ref{rescaled.model}) is $r$-persistent, the chemostat never gets clogged. 

For mathematical convenience, we express (\ref{rescaled.model}) as an equivalent {\bf modified system} (\ref{modified.model}):
\begin{equation}
\label{modified.model}
     \left\{
\begin{aligned}
    \dot{S} & = 1 - S - r\left( \frac{u}{w_0} f(S) + (1 - \frac{r}{w_0}) g(S) \right), \\
    \dot{u} & = -u + r \left( (f(S) - \alpha ) \frac{u}{w_0} +\beta (1 - \frac{r}{w_0})\right), \\
    \dot{r} & = -r \left( \alpha \frac{u}{w_0} + (g(S) - \beta)(1 - \frac{r}{w_0}) \right).
\end{aligned}
    \right.
\end{equation}  
We consider solutions in the set $\Omega_r = \left\{(S,u,r) |0\leq S \leq 1, 0\leq u \leq \frac{ 4+ \beta w_0}{4}, 0\leq r \leq w_0\right\}$.

From the analysis of the clogged state, if $E_r = (1,0,0)$ is LAS in the modified system (\ref{modified.model}) is equivalent that $E_w$ is LAS in the rescaled system (\ref{rescaled.model}). We want to show $\liminf_{t \rightarrow \infty}r(t) \geq$, which is similar to the proof of the theorem \ref{rho-strong-persistent} and theorem \ref{persistence_of_(u+w)}. We prove it by the theorem 4.13 in \cite{Smith2016DynamicalSA} and construct the following function
    \begin{equation*}
        \tau (S,u,r) = r,~ \sigma = \tau \circ \varphi
    \end{equation*}
    where, $\varphi:\mathbb{R} \times \mathbb{R}^3 \rightarrow \mathbb{R}^3$, which is the flow induced by (\ref{modified.model}). $\varphi(t,x_0)$ is the solution of (\ref{modified.model}) with the initial condition $x_0 = (S(0),u(0),r(0))$.

\begin{lem}
\label{weak-rho-perst-for-r}
    If $E_0$ and $E_w$ are unstable in the rescaled system (\ref{rescaled.model}), then $E_r$ is unstable in the modified system (\ref{modified.model}) and the system (\ref{modified.model}) is uniformly weakly $r$-persistent.
\end{lem}

\begin{proof}
For this theorem, we want to show that there exists $\eta > 0$ such that for any initial value $x_0 \in \mathring{\Omega}_r$, $\displaystyle \limsup_{t\rightarrow\infty}r(t) \geq \eta$.

We start by letting
$$ M:=  (f(1) - \alpha)\displaystyle\frac{ 4+ \beta w_0}{4w_0} + \beta, \quad \eta:=\frac{(\beta - g(1))w_0}{ 2(\beta - g(1) +2\alpha M )}.$$
Since we are assuming that both $E_0$ and $E_w$ are unstable in (\ref{rescaled.model}), we must have that 
$\beta-g(1)>0$ and $f(1)-\alpha>1$. Therefore, $M>0$ and $\eta>0$. 

For sake of contradiction, suppose that there exists $x_0 \in \mathring{\Omega}_r$ such that $\displaystyle \limsup_{t\rightarrow\infty}r(t) < \eta$. Then there exists $T_1 \geq 0$
such that $r(t)<\eta$ for all $t>T_1$, and consequently, from the $u$-equation in (\ref{modified.model}), we have that 
    \begin{equation*}
    \begin{aligned}
        \dot{u}
        & = - u + r \left( (f(S) - \alpha ) \frac{u}{w_0} +\beta (1 - \frac{r}{w_0})\right) \\
        & \leq -  u + r \left( (f(1) - \alpha)\frac{ 4+ \beta w_0}{4w_0} + \beta \right) \\
        & \leq - u + \eta M, \quad t>T_1.
    \end{aligned}
    \end{equation*}
Therefore, there exists $T_2 \geq T_1$ such that $u(t)<2 \eta M$ for all $t>T_2$. Therefore, from the $r$-equation in (\ref{modified.model}), we have that
\begin{align*}
    \dot{r} & = r \left( -\alpha \frac{u}{w_0} + (\beta-g(S))(1 - \frac{r}{w_0}) \right) \\
            & \geq r \left( - \alpha  \frac{2\eta M}{w_0} +  (\beta - g(1))  ( 1 - \frac{\eta}{w_0} ) \right) \\
            & = r \left( -\frac{\eta}{w_0}(2 \alpha M+ \beta - g(1)) +  \beta - g(1)  \right) \\
            & =  r\frac{\beta - g(1)}{2}, \quad t>T_2,
\end{align*}
which implies that $r(t) \rightarrow \infty$ as $t \rightarrow \infty$, a contradiction. Therefore, (\ref{modified.model}) is uniformly weakly $r$-persistent. \end{proof}

Next, we utilize Theorem 4.13 in \cite{Smith2016DynamicalSA} again to deduce the uniform strong $r$-persistence of (\ref{modified.model}) from the uniform weak $r$-persistence.
\begin{thm}
\label{rho-persistence-of-r}
    If $E_0$ and $E_w$ are unstable under (\ref{rescaled.model}), then (\ref{modified.model}) is uniformly strongly $r$-persistent. 
 \end{thm}
\begin{proof} The sets $B = X = \Omega_r$, $\Omega_r$ are compact. From the theorem \ref{weak-rho-perst-for-r}, we know that the modified rescaled system (\ref{modified.model}) is uniformly weakly $r$-persistent. Since $B=X$, we have that $d(\varphi_t(x),B) = 0$ for all $x \in X, r(x) >0$. In addition, $B \bigcap \left\{0< \epsilon_1 \leq \tau(x) \leq \epsilon_2 <\infty \right\}$ is closed and bounded in $\mathbb{R}^3$, hence it is compact. Finally, the set $\left\{(S,u,r)| r = 0\right\}$ is invariant under the modified system (\ref{modified.model}). Hence, all four conditions in ($\hat{\heartsuit}$) are verified. Using Theorem 4.13 in \cite{Smith2016DynamicalSA}, we conclude that (\ref{modified.model}) is uniformly strongly $r$-persistent. 
\end{proof}

\begin{thm} \label{unif_persist}
     If $E_0$ and $E_w$ of the system (\ref{rescaled.model}) are both unstable, then the system (\ref{rescaled.model}) is uniformly persistent. 
\end{thm}
\begin{proof} By Theorem \ref{rho-strong-persistent}, 
there exists $\gamma > 0, T_3>0$ such that $u(t) + w(t) > \gamma$ for all $t\geq T_3$. By Theorem \ref{rho-persistence-of-r}, 
there exists $\eta>0, T_4>0$ such that $w(t) < w_{0} - \eta$ for all $t\geq T_4$.

Then, for $t \geq T:=\max(T_3,T_4)$ we have that 
\begin{align*}
    \dot{u} & = -u + (1 - \frac{w}{w_{0}}) [(f(S) - \alpha)u + \beta w] \\
            & \geq  -u +  \frac{\eta}{w_{0}} ( - \alpha u + \beta (\gamma -u)) \\
            & \geq \frac{\eta}{w_{0}} \left( -(w_{0}/\eta+ \alpha + \beta )u(t) + \beta \gamma \right),
\end{align*}
which implies that $\displaystyle \liminf_{t\rightarrow \infty}{u(t)} \geq \frac{\beta\gamma}{w_{0}/\eta+ \alpha + \beta} > 0$. Similarly, 
$$
    \dot{w} = (1 - \frac{w}{w_{0}}) [\alpha u +(g(S)- \beta) w]\geq   \frac{\eta}{w_{0}} [\alpha (\gamma-w) -\beta w] = \frac{\eta}{w_{0}} [\alpha \gamma -(\alpha+\beta)w] $$  
 implies that $\displaystyle \liminf_{t\rightarrow \infty}{w(t)} \geq \frac{\alpha\gamma}{\alpha + \beta} > 0$. 
    
Finally, since $u(t)+w(t) \leq M:=1 + \beta w_0/4 + w_0,$ we have that
$$\dot{S} \geq 1-S -(f(S)+g(S))M,$$
hence $\displaystyle \liminf_{t\rightarrow\infty}{S(t)} \geq \hat{S}$, where $\hat{S} \in (0,1)$ is the unique root of the equation
\begin{equation}
\label{root_of_S.hat}
    1-S  = (f(S)+g(S))M.
\end{equation}
Therefore, the rescaled system (\ref{rescaled.model}) is uniformly persistent. 
\end{proof}

Under the assumptions of Theorem \ref{unif_persist}, it follows readily that the original system (\ref{dilution.model}) is also uniformly persistent. 

\section{Positive equilibria}\label{sec6}

In this section, we analyze positive equilibria and their stability. The existence of at least one positive equilibrium is guaranteed when the rescaled system (\ref{rescaled.model})
is uniformly persistent (for instance, see  Theorem 1.3.7 in \cite{MR1980821}). Although the uniform persistence is sufficient, it is not necessary for the existence 
a positive equilibrium as we will demonstrate numerically in this section. We first characterize the local stability of positive equilibria. Due to Remark \ref{remark},
we may consider the original system (\ref{dilution.model}).

The coordinates of any positive equilibrium $E_c=(S,u,w)$ satisfy
\begin{eqnarray*}
0 & = & D(w)(1- S )- u f(S) -w g(S),\\
0 & = & (f(S)-D(w) -\alpha)u +\beta w,\\
0 & = & \alpha u  +(g(S)-\beta) w.
\end{eqnarray*}
with $D(w)=\dfrac{w_0}{w_0-w}.$ 
By adding the equations, we get that $u=1-S$. From the third equation, we get that 
$$ w=\frac{\alpha u}{\beta - g(S)}=\frac{\alpha (1-S)}{\beta - g(S)}=:w(S).$$
From the second equation, we get that
$$ 0= f(S)-D(w(S)) -\alpha  +\beta \frac{w}{u} =f(S)-D(w(S)) -\alpha  +\frac{\alpha \beta }{\beta - g(S)}.$$
Multiplying by $\beta-g(S)>0$ on both sides, we find that the positive equilibria correspond to the roots of the function
$$ F(S):=(\beta - g(S))[f(S)-D(w(S))-\alpha]+\alpha \beta =(\beta - g(S))[f(S)-D(w(S))]+\alpha g(S),$$
such that $0<S<1$, $g(S)<\beta$ and $0<w(S)<w_0$.
Observe that
$$ w'(S)=-\frac{\alpha}{\beta-g(S)}+\frac{\alpha g'(S)(1-S)}{(\beta-g(S))^2}=\frac{g'(S)w-\alpha}{\beta-g(S)}.$$
Consequently,
$$ F'=-g'(f-D)+(\beta-g)(f'-D' w')+\alpha g'=-D'(w g'-\alpha)-f'(g-\beta)-g'(f-D-\alpha),$$
where we suppressed the arguments for notational convenience.

The Jacobian of (\ref{dilution.model}) at a positive equilibrium $E_c$ has the form
$$ J(E_c) =  \left[ \begin{array}{ccc} -(D+ u f'+wg') & -f & -g+D' u \\
u f' & f-D-\alpha  & \beta - D'u\\
w g' & \alpha & g-\beta  \end{array} \right].$$
We expand the determinant of $J(E)$ along the top row 
$$ \det J(E)= -(D+ u f'+w g')[(f-D-\alpha)(g-\beta) - \alpha(\beta - D' u)]$$
$$ + f[u f'(g-\beta) - w g'(\beta - D' u)]+(D'u -g)[u f'\alpha - w g'(f-D-\alpha)].$$
Using the equilibrium relation $(f-D-\alpha)(g-\beta) = \alpha \beta$
and regrouping terms, we find that 
$$ \det J(E)= D'[-(D+ u f'+w g') \alpha  u+f w g' u +u(u f'\alpha - g' w(f-D-\alpha)]$$
$$ + f'[u f(g-\beta) -g \alpha u]+ g'[-f \beta w + w g (f-D-\alpha)].$$
After cancelling the terms and substituting
$$-\alpha u = w (g-\beta), \quad - \beta w = u (f-D-\alpha), \quad Du=u f+w g,$$
we get that
$$ \det J(E_c)= D'(-\alpha D  u+ g' w D u) + f'(g-\beta)(uf+w g)+ g'(f-D-\alpha)](u f+ w g),$$
and, finally, that
$$ \det J(E_c)= Du [D'(w g'-\alpha)+f'(g-\beta)+g'(f-D-\alpha)]=-Du F'(S).$$
Since $Du>0$, it follows that $$ \sign (\det J(E_c)) =-\sign (F'(S)).$$
Hence, any positive equilibrium $E_c$ with $F'(S)<0$ is automatically unstable.

Furthermore, it turns out that any positive equilibrium $E_c$ with $F'(S)>0$ is LES. Indeed,
suppose that $F'(S)>0$ (equivalently, $\det J(E_c)<0$).
The characteristic polynomial of $J(E_c)$ has the form
\begin{eqnarray*} p(z) & = &  z^3+a_1 z^2 +a_2 z+a_3,\\
a_1 & = & \Delta_1+\Delta_2+\Delta_3 >0,\\
a_2 & = &\Delta_1 (\Delta_2+\Delta_3)+u f f'+w g g' +D'u(\alpha - w g'),\\
a_3 & = & Du [D'(\alpha - w g')+f' \Delta_3 + g'\Delta_2] >0,
\end{eqnarray*}
where 
$$\Delta_1=D+uf'+wg'>D, \quad \Delta_2=D+\alpha-f>0, \quad \Delta_3=\beta-g>0.$$
It follows that
$$ a_2>u f'(\Delta_2+\Delta_3+f)+w g' (\Delta_2+\Delta_3+g)+ D'u(\alpha - w g')$$
$$ >u f' \Delta_3+ w g' (\Delta_3+g) +D'u(\alpha - w g') $$
$$ = u f' \Delta_3+ \beta w g'  +D'u(\alpha - w g') = u f' \Delta_3+ u g' \Delta_2 +D'u(\alpha - w g').$$
since $ \Delta_2+f>\Delta_2>0$ and $(\Delta_3+g)w=\beta w = \Delta_2 u$.
Consequently, we have that
$$ a_2 > u [f' \Delta_3+  g' \Delta_2 +D'(\alpha - w g')] = \frac{a_3}{D}>0,$$
and since $a_1>\Delta_1>D$, we finally conclude that
$ a_1 a_2 > D \dfrac{a_3}{D}=a_3. $
Consequently, $a_3>0$ implies that $a_1a_2>a_3,$ which in turn implies that $a_2>0$, hence $p(z)$ is Hurwitz and $E_c$ is LES. 

The following Theorem summarizes the local stability properties of $E_c$. 
\begin{thm}
    \label{stablelas}
    A positive equilibrium $E_c=(S,u,w)$ is locally exponentially stable if $F'(S) >0,$ and unstable if $F'(S) < 0$.
\end{thm}

%the system (\ref{rescaled.model}) can be seen as the following system
%\begin{equation}
%\label{w0-infinity.model}
%     \left\{
%\begin{aligned}
%    \dot{R} & = \frac{1}{w_{0}} - R - (1 - W)(U F(R) + W G(R))  \\
%    \dot{U} & = - U + (1 - W) \left((F(R) - \alpha )U + \beta W \right) \\
%    \dot{W} & = (1 - W)[\alpha U + (G(R) - \beta)W]  
%\end{aligned}
%    \right.
%\end{equation}  
%by letting $R = \dfrac{S}{w_{0}}, U = \dfrac{u}{w_{0}}, W = \dfrac{w}{w_{0}}$, and $F(R) = f(S), G(R) = g(S)$. 

%\[
%\frac{\partial \dot{U}}{\partial W} = - (F(R) - \alpha) U + \beta (1 - 2 W)
%\]

%The theorem that we use is formulated in \cite{SMITH1982911} for a more general situation. We state a special case that is sufficient for our needs. Consider two systems in $\mathbb{R}^n$:
%\[x_0 = f(x)\]
%and
%\[y_0 = f(y) +\epsilon g(y)\]
%where $\epsilon$ is a small real parameter and $f$ and $g$ are $C^1$ functions.
%\begin{lem}[\cite{SMITH1982911}, Thm 1]
% Suppose that the eigenvalues of $J(\bar{x})$ lie in the left-half plane. Then there is an $\epsilon_2$ and a smooth family of rest points $y(\bar{\epsilon})$ for $|\epsilon| < \epsilon_0$ satisfying $\bar{y}(0) = \bar{x}$; $\bar{y}(\epsilon)$ is uniformly asymptotically stable. Moreover, if $K$ is any compact set in $B(\bar{x})$, there exists a positive number $\epsilon_1$ such that if $y_0 \in K, |\epsilon| < \epsilon_1$, then the solution $y(t, \epsilon)$ with $y(0, \epsilon) = y_0$ satisfies $\lim_{t\to\infty}| y(t, \epsilon)-\bar{y}(\epsilon)| = 0$.
%\end{lem}

\begin{thm}\label{uniqueE} If $E_0$ is unstable and $\beta>\frac{m_2}{a_2}$, then there exists a unique positive equilibrium $E_c$, which is LES.
\end{thm}

\begin{proof} We begin by observing that $\beta>\frac{m_2}{a_2}$ implies that 
$$ \beta>\frac{m_2}{a_2+1} = g(1) \geq g(S), \quad \forall S \in [0,1],$$
hence $E_w$ is unstable, and the function $w(S)=\frac{\alpha (1-S)}{\beta - g(S)} \geq 0$ is defined for all $S \in [0,1].$ Furthermore, since 
$$ \frac{\p}{\p S}\frac{a_2+S^2}{(a_2+S)^2} =\frac{2 a_2 (S-1)}{(a_2+S)^3} <0, \quad S\in (0,1),$$
we have that
$$ \frac{1}{a_2+1} \leq \frac{a_2+S^2}{(a_2+S)^2} \leq \frac{1}{a_2}, \quad \forall S \in [0,1],$$
hence 
$$ \beta > \frac{m_2}{a_2} \geq \frac{m_2(a_2+S^2)}{(a_2+S)^2}=\frac{m_2 a_2(1-S)}{(a_2+S)^2}+ \frac{m_2 S}{a_2+S}=g'(S)(1-S)+g(S), \quad \forall S \in [0,1],$$
which in turns implies that 
$$ w'(S)=\alpha\frac{g'(S)(1-S) - (\beta - g(S))}{(\beta - g(S))^2 }<0$$
for all $S \in [0,1].$ 

We rewrite the equilibrium condition $F(S)=0$ in the form
$$ D(w(S))=f(S)+\frac{\alpha g(S)}{\beta - g(S)}:=\phi(S),$$
and note that the function $\phi(S)$ is increasing on the interval $S \in [0,1]$ since $f'(S),g'(S)>0$ and $\beta - g(S)>0$. Substituting $D(w(S))=1-w(S)/w_0$ and solving for $w_0$, we obtain the equation
$$ w_0=\frac{w(S) \phi(S)}{\phi(S)-1}:=\Phi(S),$$
where $w(S)=\frac{\alpha (1-S)}{\beta - g(S)} > 0$ for all $S \in [0,1).$ The assumption that $E_0$ is unstable combined with $\beta-g(1)>0$ implies that 
$$ (f(1)-1-\alpha)(\beta -g(1))+\alpha \beta>0,$$
which translates into
$$ \phi(1)=f(1)+\frac{\alpha g(1)}{\beta - g(1)}>1.$$  Letting $S_1 =\phi^{-1}(1) \in (0,1),$ we see that $\Phi(S) \geq 0$ for all $S \in (S_1,1]$ with
$$\lim_{S \downarrow S_1} \Phi(S)=+\infty, \quad \Phi(1)=0.$$ Additionally, since $w'(S)<0$, we have that $\Phi'(S)<0$ for all $S \in (S_1,1]$.
Hence for any $w_0>0$, the equation
$w_0=\Phi(S)$ has a unique root $S^*$ such that $\Phi'(S^*)<0.$ We have that
$$ \Phi'(S) = \frac{w'(S) \phi(S)(\phi(S)-1) - w(S) \phi'(S)}{(\phi(S)-1)^2},$$
and also that 
$$ \phi(S)=D(w(S))+ \frac{F(S)}{\beta - g(S)}=\frac{w_0}{w_0-w(S)}+\frac{F(S)}{\beta - g(S)}.$$
Differentiating on both sides and substituting $F(S^*)=0$, we find that $\phi(S^*)=D(w(S^*))=\frac{w_0}{w_0-w(S^*)}$, and thus
$$ \phi'(S^*)=\frac{w_0 w'(S^*)}{(w_0-w(S^*))^2} + \frac{F'(S^*)}{\beta - g(S^*)} = \frac{\phi^2(S^*) w'(S^*)}{w_0}+\frac{F'(S^*)}{\beta - g(S^*)}.$$
At the positive equilibrium, we also have that $w_0=\Phi(S^*)=\frac{w(S^*) \phi(S^*)}{\phi(S^*)-1},$
hence
$$ \phi'(S^*)= \frac{\phi (S^*) (\phi(S^*)-1) w'(S^*)}{w(S^*)}+\frac{F'(S^*)}{\beta - g(S^*)}.$$
Therefore,
$$ \Phi'(S^*) = \frac{w'(S^*) \phi(S^*)(\phi(S^*)-1) - w(S^*) \phi'(S^*)}{(\phi(S^*)-1)^2} = 
- \frac{w(S^*) F'(S^*)}{(\beta - g(S^*))(\phi(S^*)-1)^2}, $$
which implies that
$$ \sign (\Phi'(S^*))=-\sign (F'(S^*)).$$
Since in our case,  $\Phi'(S^*)<0$, it follows that $F'(S^*)>0,$ 
and the corresponding (unique) positive equilibrium is LES. 
\end{proof}

\begin{figure}[htb]
    \centering
    \subfloat{\includegraphics[width=0.6\linewidth]{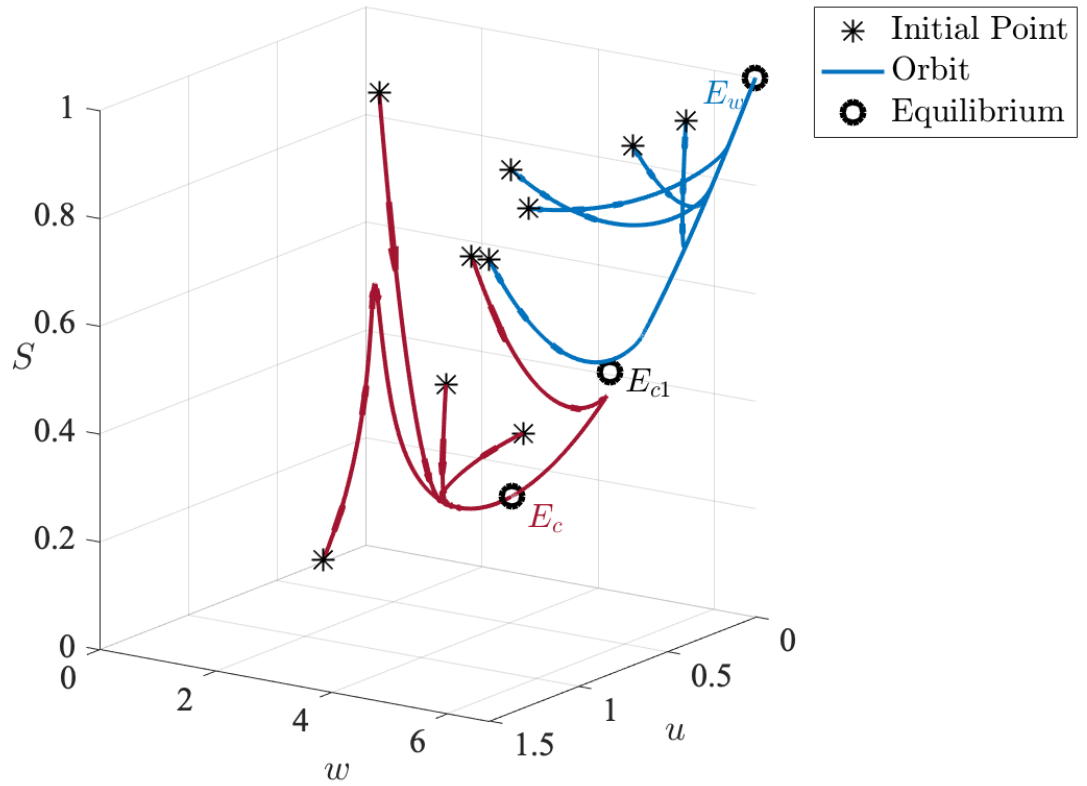}}
    \caption{In this Figure, we illustrate the bistable scenario in the case when the clogged state $E_w$ is locally stable, that is, in
    the absence of the uniform persistence. The parameter values and the growth functions are given by  $\alpha = 3.4, \beta = 0.94, w_{0} = 6.5$, $f(S) = \dfrac{5S}{2+S}$, and $g(S) = \dfrac{1.8S}{0.9+S}$, respectively.  With these parameter values, the system admits two positive equilibria: $E_{c} \sim (0.328,0.672,4.975)$ and 
    $E_{c1} \sim (0.5,0.5,5.721)$.  The equilibria $E_{c}$ and $E_{w}$ are locally stable, and $E_{c1}$ is a saddle point whose stable manifold separates the basins of attraction of $E_{c}$ and $E_{w}$, respectively. The figure shows  5 red orbits and 5 blue orbits. The red orbits correspond to initial conditions that are randomly chosen from the set $\left\{(S,u,w)|0.8 \leq S \leq 1, 0\leq u \leq 1+\frac{\beta w_0}{4}, 0.9 w_0 \leq w < w_0\right\}$, and all
    red orbits are attracted by $E_{c}$.  The blue orbits correspond to initial conditions that are randomly chosen from the set  $\left\{(S,u,w)|0.8 \leq S \leq 1, 0\leq u \leq 1+\frac{\beta w_0}{4}, 0.9 w_0 \leq w < w_0\right\}$, and  all blue orbits are attracted by $E_{w}$. }
    \label{Bistable_without_UP}
\end{figure}

%There is an example that $E_c$ is locally stable shown in the plot below. The red curve represents the sample path image of the rescaled system, the blue curve represents the sample path image of the original system. The three subgraphs represent the nutrient-time, normal speices-time, and adherent species-time relationship in turn. 

%\begin{figure}[htb]
%    \centering
%    \subfloat{\includegraphics[width=0.48\linewidth]{Dilutionchange/graph/Ec/S_t.pdf}}
%    \hfill
%   \subfloat{\includegraphics[width=0.48\linewidth]{Dilutionchange/graph/Ec/u_t.pdf}} \\
    
%    \subfloat{\includegraphics[width=0.48\linewidth]{Dilutionchange/graph/Ec/w_t.pdf}}    
%    \caption{The simulation of the trajectory of the original system (\ref{rescaled.model}). We assume that the initial value of the system (\ref{rescaled.model}) is $(S(0), u(0), w(0)) = (0.2, 1, 3) \in \Omega$. The parameter setting $\alpha = 3.5, \beta = 0.95, w_{0} = 6.8$ and two growth rate functions are given by $f(S) = \dfrac{5S}{2+S}$ and $g(S) = \dfrac{1.8S}{0.9+S}$.}
%    \label{washout}
%\end{figure}

\begin{thm}\label{nonuniqueE} If $\beta<g(1)$ and $w_0 \gg 1$ is sufficiently large, then there exist
at least two positive equilibria of (\ref{rescaled.model}).
\end{thm}

\begin{proof} We use a similar approach as in the previous Theorem. Letting $S_2:=g^{-1}(\beta) \in (0,1)$, we see that the function
$\phi(S)=f(S)+\frac{\alpha g(S)}{\beta - g(S)}$ is increasing on the interval $S \in [0,S_2)$ with $\phi(0)=0$ and 
$$\lim_{S \uparrow S_2} \phi(S)=\lim_{S \uparrow S_2} w(S)=+\infty.$$ 
We let $S_1 =\phi^{-1}(1) \in (0,S_2).$ Using the same function $\Phi(S)$ as in the proof of Theorem \ref{uniqueE}, we find that $\Phi(S)>0$
for all $S \in (S_1,S_2)$ and
$$ \lim_{S \downarrow S_1} \Phi(S)=\lim_{S \uparrow S_2} \Phi(S)=+\infty.$$
Therefore, $\Phi(S)$ attains its minimal value 
$\Phi(\hat S)$ at some $\hat S \in (S_1,S_2).$ Consequently, for any $w_0 > \Phi(\hat S),$
the equation $w_0=\Phi(S)$ admits two distinct roots $S_1^*$ and $S_2^*$ such that
$$ S_1<S_1^* < \hat S < S_2^*<S_2.$$
Therefore, (\ref{rescaled.model}) admits
at least two positive equilibria. \end{proof}

\begin{rem} The assumption that $g(1)<\beta$ in Theorem \ref{nonuniqueE} corresponds to $E_w$ being locally stable. Having two additional positive equilibria corresponding to the values $S_1^*<S_2^*$ would result in a bistable scenario, where there will be initial conditions for which the chemostat becomes clogged in finite time and other initial conditions for which the chemostat operates at a stable positive equilibrium. Specifically, the equilibrium with $S=S_1^*$ will have $\Phi'(S_1^*)<0,$ and hence be locally stable, while the equilibrium with $S=S_2^*$ will have $\Phi'(S_2^*)>0,$ and hence be unstable. 
\end{rem}

\begin{figure}[htb]
    \centering
    \subfloat{\includegraphics[width=0.6\linewidth]{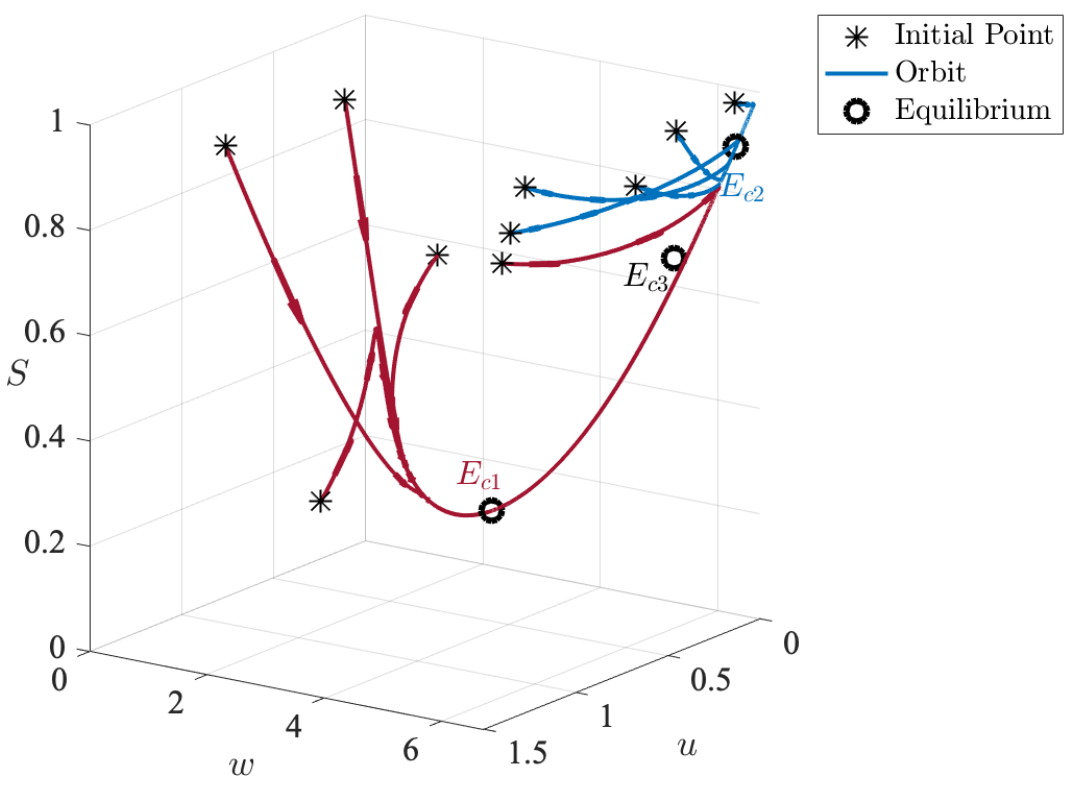}}
    \caption{In this Figure, we illustrate the bistable scenario in the uniformly persistent case. The parameter values and the growth functions are given by  $\alpha = 3.3, \beta = 0.95, w_{0} = 6.5$, $f(S) = \dfrac{5S}{2+S}$, and $g(S) = \dfrac{1.8S}{0.9+S}$, respectively.  With these parameter values, the system admits three distinct positive equilibria: $E_{c1} \sim (0.282,0.718,4.544)$,
    $E_{c2} \sim (0.948,0.052,6.442)$, and 
    $E_{c3} \sim (0.712,0.288,6.153)$.  The equilibria $E_{c1}$ and $E_{c2}$ are locally stable, and $E_{c3}$ is a saddle point whose stable manifold separates the basins of attraction of $E_{c1}$ and $E_{c2}$, respectively. The figure shows  5 red orbits and 5 blue orbits. The red orbits correspond to initial conditions that are randomly chosen from the set $\left\{(S,u,w)|0\leq S < 0.9, 0\leq u \leq 1+\frac{\beta w_0}{4}, 0\leq w < 0.95 w_0\right\}$, and all
    red orbits are attracted by $E_{c1}$.  The blue orbits correspond to initial conditions that are randomly chosen from the set  $\left\{(S,u,w)|0.9 \leq S < 1, 0\leq u \leq 1+\frac{\beta w_0}{4}, 0.96 w_0 \leq w < w_0\right\}$, and  all blue orbits are attracted by $E_{c2}$.  }
    \label{Bistable_with_UP}
\end{figure}

\section{Discussion}\label{sec7}
In this paper, we have presented a novel mathematical model of a chemostat where the effective dilution rate varies due to the microbial biofilm growth.  The main novelty of this new model is that some solutions can reach the clogged state in finite time corresponding to the eventual failure of the bioreactor. Importantly, we have shown that the chemostat may become clogged regardless of the total volume of its reactor chamber $w_0$. The parameters that determine the stability of the clogged state are exclusively related to the difference $g(1)-\beta$ between the specific growth rate at maximal resource concentration $g(1)$ and the sloughing off rate $\beta$ of the biofilm.

We have also shown that the qualitative behavior of the new model is much richer than that of a traditional chemostat with wall growth. In particular, we provided sufficient conditions for two distinct multistable regimes of the chemostat. In the first regime, both the washout and the clogged states are unstable, hence the model exhibits uniform persistence of the microbial species, but  there exist two distinct positive stable equilibria, each with its own nonempty basin of attraction.
In the second regime, the washout state is unstable, but while the clogged state is stable, the chemostat also admits  a stable positive equilibrium, so that for some initial conditions, the chemostat will reach the clogged state in finite time,  and for others, the chemostat will operate stably at a positive equilibrium. We have also shown that the number of positive equilibria is sensitive to the total volume of the reactor chamber $w_0$, which is a property that may have implications for chemostat design.

We have performed a complete stability analysis of both boundary and positive equilibria and showed that the positive equilibria do not undergo Hopf bifurcations,  thus making sustained oscillations unlikely, although a rigorous proof of this claim remains elusive. 

The chemostat model that we described in this paper is a first step towards understanding the combination of factors that lead to decreased performance of a chemostat up to a complete failure of the device due to bio-clogging. The model suffers from many simplifications such as the assumption of well-mixedness and the general lack of an explicit spatial description of the interface between the biofilm and the fluid chamber. These are the questions to be addressed in subsequent modifications of the model.

%\
%\begin{table}[htbp]
%  \centering
%  \begin{tabular}{c|c}
%\hline
% \tabincell{c}{$ g(1) - \beta < \displaystyle\frac{\alpha \beta}{f(1) - D_0 - \alpha}$ \\ $ g(1) - \beta < 0$ } & \tabincell{c}{$E_0$ is GAS.}  \\ 
%\hline
% \tabincell{c}{ $f(1) - D_0 - \alpha > \displaystyle\frac{\alpha \beta}{g(1) - \beta} $ \\ $g(1) - \beta < 0$} & \tabincell{c}{The system is uniformly persistent.}  \\
%\hline
% \tabincell{c}{ $g(1) - \beta > 0$} & \tabincell{c}{$E_b$ is LAS, \\ The system reached the clogged state in finite time.} \\
%\hline
%\end{tabular}
%%--------------------------------------------------------------------------------------------------
%  \caption{Summary}
%\end{table}

%As a final touch, we point out that the Jacobian matrices of the original system
%(\ref{dilution.model}) and the rescaled system (\ref{rescaled.model}) at any positive equilibrium $E = (S^∗, u^∗, w^∗) $are related via $DG(E) = kDF(E)$, where  $k = 1 - \dfrac{w^∗}{w_{0}}$ is a positive scalar, so $DG(E)$ is Hurwitz if and only if $DF(E)$ is Hurwitz. Not surprisingly, the positive equilibria have the same stability properties in both systems.
%\newpage

\pagestyle{empty}
\bibliographystyle{plain}
\bibliography{Manuscript.bib}

\end{document}